\documentclass{article}
\usepackage{arxiv}
\usepackage{amsmath,amsfonts}
\usepackage{algorithmic}
\usepackage{algorithm}
\usepackage{array}
\usepackage[caption=false,font=normalsize,labelfont=sf,textfont=sf]{subfig}
\usepackage{textcomp}
\usepackage{stfloats}
\usepackage{subcaption}
\usepackage{url}
\usepackage{multirow}
\usepackage{verbatim}
\usepackage{graphicx}
\usepackage{cite}
\usepackage{amsmath}
\usepackage{amssymb}
\usepackage{amsfonts}
\usepackage{amsthm}
\newtheorem{theorem}{Theorem}
\newtheorem{proposition}[theorem]{Proposition}
\newtheorem{corollary}[theorem]{Corollary}
\usepackage[utf8]{inputenc} 
\usepackage[T1]{fontenc}    
\usepackage{hyperref}       
\usepackage{url}            
\usepackage{booktabs}       
\usepackage{amsfonts}       
\usepackage{nicefrac}       
\usepackage{microtype}      
\usepackage{lipsum}
\usepackage{graphicx}
\graphicspath{ {./images/} }

\title{TADP-RME: A Trust-Adaptive Differential Privacy Framework for Enhancing Reliability of Data-Driven Systems}

\author{
Labani Halder \\
Indian Statistical Institute Kolkata \\
203 B.T. Road, Kolkata-700108, India \\
\texttt{labanihalder1@gmail.com} \\
\And
Payel Sadhukhan \\
Army Institute of Management \\
Plot No III/B-11, Action Area III, New Town, Kolkata-700160, India \\
\texttt{payel0410@gmail.com} \\
\And
Sarbani Palit \\
Indian Statistical Institute Kolkata \\
203 B.T. Road, Kolkata-700108, India \\
\texttt{} \\
}

\begin{document}
\maketitle

\begin{abstract}
Ensuring reliability in adversarial settings necessitates treating privacy as a foundational component of data-driven systems. While differential privacy and cryptographic protocols offer strong guarantees, extant schemes rely on a fixed privacy budget, leading to a rigid utility–privacy trade-off that fails under heterogeneous user trust. Moreover, noise-only DP preserves geometric structure, which inference attacks exploit, causing privacy leakage - a system failure mode. We propose \textbf{TADP-RME} (Trust-Adaptive DP with Reverse Manifold Embedding), a framework that enhances reliability in adversarial conditions with varying levels of user trust. TADP-RME's modus operandi introduces an inverse trust score \(\mathcal{T} \in [0,1]\) to attenuate the privacy budget, enabling smooth, interpretable transitions between high-utility (low-privacy) and low-utility (high-privacy) requirements. It further applies \textit{Reverse Manifold Embedding} (RME), a nonlinear transformation to jumble the local proximity relationships and accentuate inversion ambiguity (despite preserving formal \((\varepsilon,\delta)\)-DP guarantees in post-processing). 
Theoretical analysis and experimental outcomes show that TADP-RME improves the privacy–utility trade-off, reducing attack success rates by up to $3.1\%$ without significant utility loss. It consistently outperforms existing methods against inference attacks, establishing a unified approach to guarantee system reliability under adversarial constraints.
\end{abstract}


\section{Introduction}
\noindent The rapid adoption of data-driven systems in sensitive domains such as healthcare, finance, and personalized services has mandated privacy as a fundamental pillar of systems' reliability \cite{data-driven1,data-driven2}. Differential Privacy (DP) has emerged as a principled framework for protecting individual data and improving the reliability of data-driven systems, offering strong guarantees that the inclusion or exclusion of a single record does not significantly affect the outcome of an analysis \cite{dwork2006, dwork2014book}. However, recent empirical studies have shown that DP alone does not fully eliminate privacy leakage, especially when models preserve structural and statistical properties that can be exploited by inference attacks \cite{jayaraman2020evaluating, shokri2017membership, nasr2019comprehensive, carlini2021extracting, fredrikson2015model}. From a system reliability perspective, privacy leakage can be interpreted as a failure event that compromises the dependability of data-driven systems \cite{privacy-rel1,privacy-rel2}. 
Consequently, designing privacy-preserving mechanisms can be viewed as a reliability engineering problem, in which the objective is to minimize the failure probability under adversarial conditions.

The attacks often stem from vulnerabilities arising from residual geometric footprints in the perturbed data, such as pairwise distances, clustering structure, and neighborhood configuration. These aspects constitute key limitations of conventional DP mechanisms: \textit{while they perturb data values, they do not explicitly disrupt geometric structure, which remains a critical source of privacy leakage}. Additionally, conventional DP mechanisms rely on a fixed privacy budget $\varepsilon$, enforcing a uniform trade-off between privacy and utility across all users and contexts. In practice, however, privileges are rarely uniform. Different users, applications, or operational settings often demand varying levels of privacy protection. Applying a single privacy budget across heterogeneous trust substrates can lead to suboptimal outcomes, either unnecessarily degrading utility or failing to provide sufficient privacy. Recent work has explored adaptive and personalized variants of differential privacy to address this limitation \cite{jorgensen2015pdp, ebadi2016pdp}. While these approaches introduce flexibility in noise calibration, they typically operate within the same noise-injection paradigm and do not modify the underlying layout of the data. As a result, they remain vulnerable to modern inference attacks that exploit geometric and statistical structure, particularly in correlated data settings. In particular, distance-based and representation-based attacks can exploit residual neighborhood structure even after noise-based data obfuscation. 

To address these limitations, we introduce {TADP-RME}, a Trust-Adaptive Differential Privacy framework with Reverse Manifold Embedding, whose objective is to improve reliability under heterogeneous adversarial conditions. The work has two key objectives -- i] adaptive privacy control, and, ii] structural transformation to enhance robustness against inference attacks. Accordingly, our framework consists of two components. First, we introduce an inverse trust metric, $\mathcal{T} \in [0,1]$ that quantifies risk: $\mathcal{T}=0$ signifies a highly trusted, low-risk context and $\mathcal{T}=1$ denotes a completely untrusted, high-risk environment. We use this information to adaptively determine the privacy budget $\varepsilon_{\mathcal{T}}$. This metric enables a smooth, mathematically interpretable privacy–utility trade-off. Second, we propose a nonlinear geometric transformation mechanism, Reverse Manifold Embedding (RME) designed to disrupt local proximity relationships and reduce the effectiveness of geometry-based inference attacks. 
RME maps data into a higher dimensional space using a nonlinear periodic embedding that intentionally distorts neighborhood structure. In contrast to classical manifold learning techniques that preserve local geometry, RME intentionally distorts neighborhood relationships, such that, nearby points in the original space may become distant after transformation — thereby increasing ambiguity in inverse mapping. This design is inspired by nonlinear manifold transformation (e.g, Swiss-roll type function), which reorder proximity relationships and increase ambiguity in inverse mapping and enhance overall reliability of systems. 

In this work, we provide a comprehensive theoretical analysis of the proposed framework. This includes formal privacy guarantees, information-theoretic bounds on information leakage as a system failure mode within a reliability framework, and complexity analysis of inversion under geometric deformation. We further validate the approach through experiments on benchmark datasets and contemporary competing methods. Results show that TADP-RME achieves a favorable privacy–utility trade-off while improving system reliability against adversarial inference. In particular, it outperforms standard differential privacy mechanisms and personalized baselines. 

The main contributions of this work are as follows:
\begin{itemize}

\item 
We propose a trust-adaptive framework in which an inverse trust score, $\mathcal{T} \in [0,1]$, governs the privacy budget, enabling a flexible and interpretable privacy–utility trade-off (transitioning beyond the fixed-budget Differential Privacy).

\item 
We introduce Reverse Manifold Embedding (RME), a nonlinear transformation to disrupt the structural dependencies, thereby reducing susceptibility to geometry-based inference attacks.



\item 
Empirical results show that TADP-RME achieves improved privacy–utility trade-offs and enhanced robustness compared to classical and personalized Differential Privacy baselines.

\end{itemize}

The rest of the paper is organized as follows. Section II reviews the related work, setting the context for the study. Building on this, Section III presents the problem formulation. Section IV then introduces the proposed TADP-RME framework. Section V provides the theoretical analysis, and Section VI describes the experimental setup along with the corresponding results. Finally, Section VII concludes the article.

\section{Related Work}
\noindent Differential Privacy (DP) provides a rigorous and widely adopted framework for protecting sensitive data and improving the reliability of data-driven systems, offering formal guarantees that limit the influence of any individual record on the output of a computation \cite{dwork2006, dwork2014book}. Classical mechanisms, such as the Laplace and Gaussian mechanisms, achieve $(\varepsilon, \delta)$-DP by injecting calibrated noise into query outputs or data representations. Over time, DP has been extended to a wide range of settings, including local differential privacy, distributed learning, and deep learning. More recent formulations, such as Gaussian Differential Privacy \cite{dong2022gdp}, further refine the interpretation and analysis of privacy guarantees. Despite these advances, practical deployments of DP often reveal a gap between theoretical guarantees and empirical privacy leakage, which can be interpreted as system failure events in real-world machine learning systems \cite{jayaraman2019evaluating}. In particular, traditional DP mechanisms rely on a fixed privacy budget $\varepsilon$, enforcing a uniform privacy–utility trade-off across all users and contexts. However, real-world data access is inherently heterogeneous, with varying trust levels and privacy requirements across users and applications. This mismatch can lead to either excessive utility degradation or insufficient privacy protection. To address this limitation, personalized and adaptive variants of differential privacy have been proposed \cite{jorgensen2015pdp, ebadi2016pdp}. Personalized Differential Privacy (PDP) enables user-specific privacy budgets, while adaptive approaches dynamically adjust noise levels based on contextual factors or data characteristics. From a reliability perspective, these approaches do not explicitly model or minimize failure probability under adversarial conditions, limiting their effectiveness in reliability-critical systems. Although these methods improve flexibility, they largely operate within the standard noise-injection paradigm and do not explicitly  modify the underlying data representation. As a result, they remain vulnerable to inference attacks that exploit structural and statistical properties of the data, particularly in high-dimensional settings \cite{shokri2017membership, nasr2019comprehensive}. Importantly, these approaches perturb data values but do not explicitly disrupt geometric relationships such as pairwise distances or neighborhood structure, which remain key signals for many inference attacks. Beyond noise-based mechanisms, transformation-based approaches have been explored to enhance privacy \cite{aggarwal2008privacy, bingham2001random, liu2019random}. Techniques such as random projection, dimensionality reduction, and feature perturbation aim to obscure sensitive information by modifying feature representations while preserving utility. Hashing-based methods, including locality-sensitive hashing (LSH) \cite{indyk1998lsh}, similarly transform data while maintaining approximate similarity. However, many of these approaches either lack formal differential privacy guarantees or incur significant utility loss, which may negatively impact system reliability in practical deployments. Moreover, a large class of such transformations are linear or approximately distance-preserving, and therefore fail to sufficiently disrupt neighborhood relationships that can be exploited by adversaries. Recent work has highlighted the vulnerability of privacy-preserving mechanisms to modern inference attacks, including membership inference \cite{shokri2017membership}, model inversion \cite{fredrikson2015model}, and data extraction attacks \cite{carlini2021extracting}. These attacks exploit statistical patterns, model outputs, and learned representations to recover sensitive information, effectively acting as failure mechanisms in privacy-preserving systems. Notably, many of these attacks rely on geometric consistency and representation similarity rather than exact data values, revealing fundamental limitations of noise-only protection mechanisms. Empirical studies further demonstrate that even differentially private models can leak sensitive information in practical settings, particularly when structural patterns remain partially preserved \cite{jayaraman2019evaluating}. To resolve these risks, recent approaches have explored adversarial training, representation learning, and hybrid privacy mechanisms that aim to remove sensitive information from learned representations. Additionally, training-based privacy mechanisms such as DP-SGD \cite{abadi2016dpsgd} introduce noise during model optimization to provide end-to-end privacy guarantees. However, these methods operate at the training level rather than directly modifying input data representations, and are therefore complementary to data-level privacy mechanisms. Despite these advances, existing methods typically address either stochastic privacy guarantees or structural robustness in isolation, but not both simultaneously. A unified framework that jointly incorporates adaptive privacy control and explicit structural distortion remains an open challenge, particularly from a system reliability perspective where both formal guarantees and empirical robustness must be jointly ensured \cite{cummings2023dp}. In contrast, the proposed TADP-RME framework integrates trust-adaptive privacy control with nonlinear geometric distortion, explicitly targeting both value-based and structure-based leakage. By dynamically adjusting the privacy budget and disrupting geometric relationships through nonlinear embedding, the proposed approach bridges formal differential privacy guarantees with improved empirical robustness against inference attacks and enhanced system reliability under adversarial conditions.

\section{Problem Statement}
\noindent Modern data-driven systems operate under heterogeneous trust requirements, where different users, entities, or applications demand varying levels of privacy protection and system reliability. 

\begin{figure}[ht]
    \centering
    \includegraphics[width=1\linewidth]{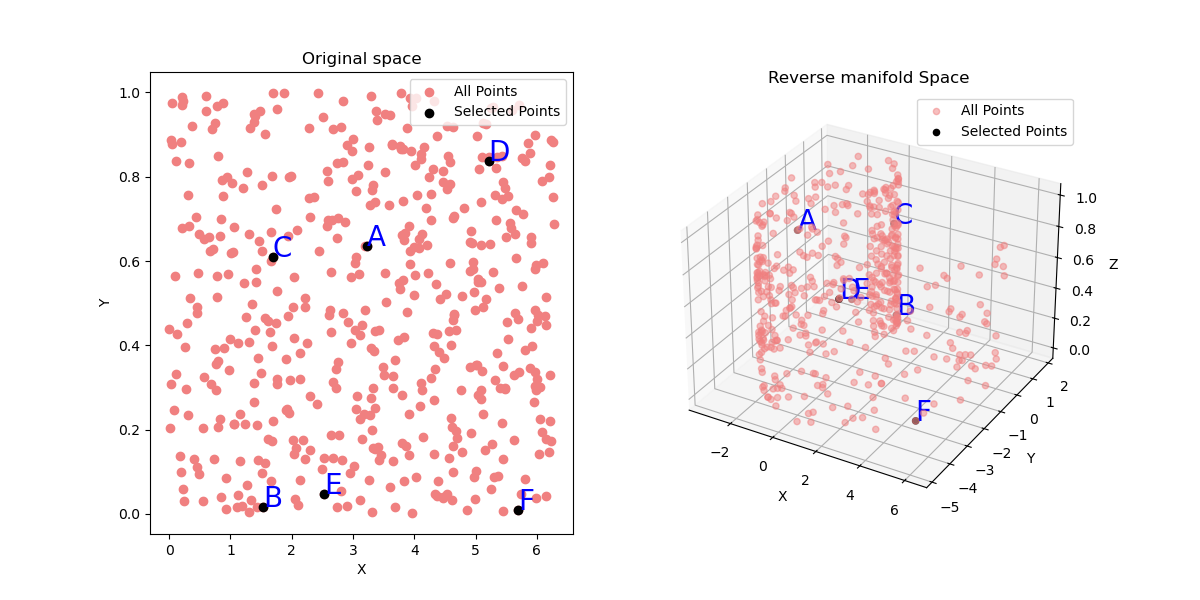}
    \caption{Motivation for Reverse Manifold Embedding (RME). (Left) Original data exhibits tightly clustered points with strong proximity relationships. (Right) After nonlinear twisting into a Swiss-roll–like manifold, nearby samples are separated and redistributed, distorting local structure and introducing ambiguity in neighborhood relationships.}
    \label{fig:swiss}
\end{figure}

Formally, given a dataset $X \in \mathbb{R}^{n \times d}$ and an inverse trust score $\mathcal{T} \in [0,1]$ (where $\mathcal{T}=0$ corresponds to maximum utility for trusted queries and $\mathcal{T}=1$ corresponds to maximum privacy for untrusted queries), the objective is to construct a mechanism $\mathcal{M}_{\mathcal{T}}$ such that
\begin{itemize}
\item $\mathcal{M}_{\mathcal{T}}$ satisfies $(\varepsilon_{\mathcal{T}}, \delta)$-differential privacy
\item the privacy budget $\varepsilon_{\mathcal{T}}$ is adaptively controlled by $\mathcal{T}$
\item the privacy–utility trade-off evolves smoothly with respect to $\mathcal{T}$
\item the mechanism is robust against reconstruction and inference attacks, thereby improving system reliability under adversarial conditions that utilize both statistical and geometric properties of the data.

\end{itemize}

Extant approaches fall short of addressing this problem in a unified manner. Fixed $\varepsilon$ DP mechanisms lack adaptability to heterogeneous trust scenarios, while personalized DP approaches typically require per-user tuning without providing structural protection against inference attacks. Moreover, noise-only mechanisms perturb data values but do not explicitly deform geometric relationships such as distances and neighborhood structure. As a result, latent correlations and proximity patterns can still be manipulated to recover sensitive information.

To this end, we address the following research question.
\\ \emph{How can we design a reliable privacy mechanism that renders an appropriate privacy-utility trade-off based on trust, while preserving formal differential privacy guarantees and diminishing leakage arising from both data values and geometric structure?}

We propose a trust-adaptive differential privacy framework augmented with nonlinear geometric distortion, enabling flexible privacy control, enhanced resistance to inference attacks, and improved system reliability under adversarial conditions. We tackle the problem as designing a mechanism that minimizes adversarial failure probability while preserving utility under heterogeneous trust conditions.

\section{Methodology}
\subsection{Framework Overview}
\begin{algorithm}[t]
\caption{TADP-RME: Trust-Adaptive Differential Privacy with Geometric Distortion}
\label{alg:tadp}
\begin{algorithmic}[1]

\REQUIRE Dataset $X \in \mathbb{R}^{n \times d}$, trust score $\mathcal{T} \in [0,1]$, privacy bounds $\varepsilon_{\min}, \varepsilon_{\max}$, failure probability $\delta$, sensitivity $\Delta_2$, distortion parameter $\alpha$

\ENSURE Protected representation $Z_{\mathcal{T}}$

\STATE \textbf{// Step 1: Trust-Adaptive Privacy Budget}
\STATE Compute $\varepsilon_{\mathcal{T}} = \varepsilon_{\max} - \mathcal{T}(\varepsilon_{\max} - \varepsilon_{\min})$

\STATE \textbf{// Step 2: Noise Calibration}
\STATE Compute variance: $
\sigma_{\mathcal{T}}^2 = \frac{2 \Delta_2^2 \log(1.25/\delta)}{\varepsilon_{\mathcal{T}}^2}
$

\STATE \textbf{// Step 3: Gaussian Perturbation}
\STATE Sample noise $N \sim \mathcal{N}(0, \sigma_{\mathcal{T}}^2 I)$
\STATE Compute $Y_{\mathcal{T}} = X + N$

\STATE \textbf{// Step 4: Geometric Distortion (RME)}
\FOR{each data point $x_i \in Y_{\mathcal{T}}$}
    \STATE Compute
    $z_i = \left(x_i \cos(\alpha x_i), \; x_i \sin(\alpha x_i)\right)$
    
\ENDFOR

\STATE \textbf{// Step 5: Output}
\STATE Return $Z_{\mathcal{T}} = \{z_i\}_{i=1}^{n}$

\end{algorithmic}
\end{algorithm}
\noindent We propose {TADP-RME}, a two-stage framework that integrates trust-adaptive differential privacy with nonlinear geometric transformation. The objective is two fold: i] provide a trust-adaptive privacy budget, and, ii] control information leakage while improving resilience against adversarial inference. Given input data $X \in \mathbb{R}^{n \times d}$ and a trust score $\mathcal{T} \in [0,1]$, the mechanism produces a protected representation
\begin{equation}
Z_{\mathcal{T}} = \varphi(\mathcal{M}_{\text{DP}}(X; \mathcal{T}))
\end{equation}
where $\mathcal{M}_{\text{DP}}$ denotes a trust-adaptive Gaussian mechanism and $\varphi$ represents a nonlinear transformation.  

\subsection{Trust-Adaptive Gaussian Mechanism}
\noindent To enable adaptive control, we define an inverse trust metric $\mathcal{T} \in [0,1]$ that governs protection strength. A value of $\mathcal{T}=0$ corresponds to a trusted setting with zero (minimum) intervention, while $\mathcal{T}=1$ represents an untrusted condition requiring full (maximum) protection. The trust-dependent privacy budget is defined as $\varepsilon_{\mathcal{T}} = \varepsilon_{\max} - \mathcal{T}(\varepsilon_{\max} - \varepsilon_{\min}).$
This formulation provides a continuous transition between utility and protection regimes. The corresponding noise variance is given by $\sigma_{\mathcal{T}}^2 = \frac{2\Delta_2^2 \log(1.25/\delta)}{\varepsilon_{\mathcal{T}}^2}.$
The perturbed representation is computed as
\begin{equation}
Y_{\mathcal{T}} = X + \mathcal{N}(0, \sigma_{\mathcal{T}}^2 I).
\end{equation}
This mechanism adjusts noise intensity according to trust, enabling a controlled trade-off between data utility and resistance to inference.

\subsection{Reverse Manifold Embedding (RME)}
The second stage of this research is motivated to reduce leakage originating from geometric residuals. To this end, we introduce {Reverse Manifold Embedding (RME)}, a nonlinear periodic mapping from $\mathbb{R}^d$ to $\mathbb{R}^{2d}$
\begin{equation}
\varphi(x_i) = \big(x_i \cos(\alpha x_i), \; x_i \sin(\alpha x_i)\big).
\end{equation}
The objective of this mapping is to distort the proximity relationships, causing nearby points in the original space to become separated after transformation. Unlike conventional manifold learning, which maintains local geometry, this approach intentionally alters spatial relationships through nonlinear interactions and dimensional expansion. The transformation introduces -- i] nonlinear feature interactions that break linear dependencies, ii] dimensional expansion that increases representation complexity, and iii] ambiguity in inversion (for a trespasser) due to periodicity and non-injectivity. The parameter $\alpha$ controls distortion strength and can be tuned according to trust conditions. The design enables simultaneous preservation of adaptive formal guarantees and reduction of exploitable structure. The distortion parameter $\alpha$ controls the strength of the transformation and can be fixed or adaptively adjusted based on $\mathcal{T}$. This decoupled design enables the simultaneous achievement of formal guarantees and structural robustness. Figure~(\ref{fig:swiss}) illustrates how geometric transformation can disrupt local structure beyond what noise alone can achieve. This behavior is conceptually inspired by nonlinear manifold distortions like the Swiss-roll, which alter geometric relationships without preserving local neighborhoods. The proposed framework jointly combines stochastic perturbation and geometric distortion to address complementary sources of privacy leakage arising from both data values and structural relationships.
\subsection{Reliability Interpretation of TADP-RME}

In this part, we analyze the proposed framework from a reliability engineering perspective. Privacy leakage is modeled as a failure event, where successful inference attacks indicate a compromise of system confidentiality. The trust score $\mathcal{T}$ acts as a risk exposure parameter that determines the level of protection required under different operating conditions. Lower values correspond to controlled environments, while higher values indicate increased exposure to adversarial threats. We define a reliability function as $R = 1 - P_{\text{attack}}$, where $P_{\text{attack}}$ denotes the probability of successful adversarial inference. In practice, this corresponds directly to the empirical privacy score used in the evaluation, ensuring consistency between theoretical interpretation and experimental measurement. This formulation follows classical reliability theory, where reliability represents the probability of operation without failure. In this context, adversarial success corresponds to failure, and $R$ quantifies the system’s ability to resist such outcomes. Higher values indicate stronger protection against inference-based threats. The proposed framework improves reliability through two complementary mechanisms
\begin{itemize}
\item {\textit{Adaptive noise injection reduces the likelihood of successful inference.}}
\item {\textit{Geometric transformation increases ambiguity in reconstruction.}}
\end{itemize}
Together, these components reduce failure probability while preserving functional utility. This formulation establishes a direct link between adversarial risk and reliability, enabling quantitative evaluation of system performance under varying trust conditions.

\section{Theoretical Analysis of TADP-RME Framework}
\noindent In this section, we analyze the theoretical foundations of the proposed framework from three complementary perspectives. First, we prove that the trust-adaptive noise calibration strictly satisfies formal $(\varepsilon, \delta)$-differential privacy, and we quantify its impact on statistical distinguishability \cite{cover2006elements}. Second, we analyze the computational complexity of inverting the Reverse Manifold Embedding (RME), showing that its nonlinear dimensional expansion leads to combinatorial growth in the inversion search space under naive pairing assumptions. Finally, we provide information-theoretic bounds to quantify how the framework limits data leakage. Together, these analyses demonstrate that the proposed method separates formal guarantees from structural protection.

\subsection{Formal Privacy Guarantees with Trust Adaptation}
\subsubsection{Differential Privacy Guarantees}
\noindent To ensure bounded global sensitivity prior to noise injection, we assume the input records $x_i \in X$ are projected onto an $L_2$ ball of radius $C$, such that $\|x_i\|_2 \leq C$. Consequently, the $L_2$ sensitivity is bounded by $\Delta_2 \leq C$.

\begin{theorem}
Let $\mathcal{M}_{\mathcal{T}}$ be the TADP-RME mechanism with trust score $\mathcal{T} \in [0,1]$. For any adjacent datasets $D, D' \in \mathbb{R}^{n \times d}$ differing in at most one record, and for any measurable subset $S \subseteq \mathbb{R}^{n \times 2d}$, we have
\begin{equation}
\mathbb{P}[\mathcal{M}_{\mathcal{T}}(D) \in S] 
\leq 
e^{\varepsilon_{\mathcal{T}}} \mathbb{P}[\mathcal{M}_{\mathcal{T}}(D') \in S] + \delta
\end{equation}

where the trust-adaptive privacy budget is defined as $
\varepsilon_{\mathcal{T}} = \varepsilon_{\max} - \mathcal{T}(\varepsilon_{\max} - \varepsilon_{\min})$ and the corresponding Gaussian noise variance satisfies $\sigma_{\mathcal{T}}^2 = \frac{2\Delta_2^2 \log(1.25/\delta)}{\varepsilon_{\mathcal{T}}^2}
$
\end{theorem}

\begin{proof}
The result follows from the compositional structure of the mechanism. The trust-adaptive mechanism $\mathcal{M}_{\text{TCDP}}(x) = x + \mathcal{N}(0, \sigma_{\mathcal{T}}^2 I_d)$ satisfies $(\varepsilon_{\mathcal{T}}, \delta)$-differential privacy when $\sigma_{\mathcal{T}}$ is calibrated according to the Gaussian mechanism \cite{dwork2014book}. The reverse manifold embedding $\varphi: \mathbb{R}^d \rightarrow \mathbb{R}^{2d}$ is a deterministic mapping. By the post-processing property of differential privacy, applying $\varphi$ to a differentially private output does not weaken the guarantee. Therefore, the full mechanism $\mathcal{M}_{\mathcal{T}} = \varphi \circ \mathcal{M}_{\text{TCDP}}$ satisfies $(\varepsilon_{\mathcal{T}}, \delta)$-differential privacy.
\end{proof}

\subsubsection{Statistical Distinguishability Analysis}
\begin{figure}[ht]
    \includegraphics[width=1\linewidth]{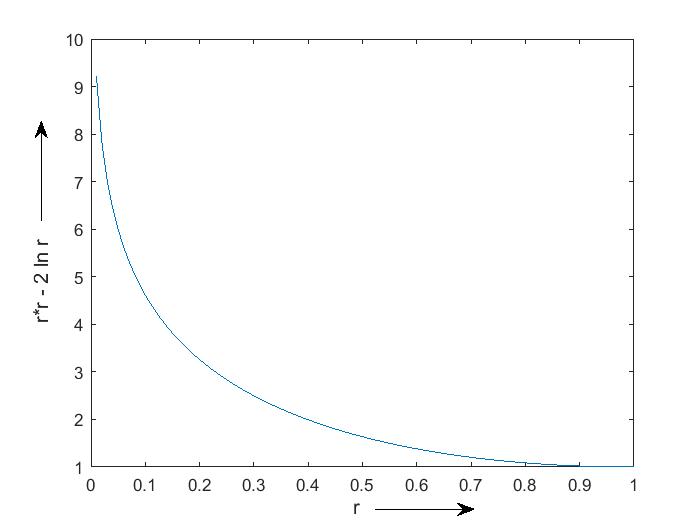}
      \caption{Behavior of $r^2 - 2\ln(r)$ for $r < 1$. The function remains strictly positive across this range, confirming inequality (\ref{eqno:r}) and guaranteeing a positive KL divergence between the corresponding distributions.}
\end{figure}
\noindent Beyond formal guarantees, we analyze how trust-adaptive noise influences distinguishability of outputs across varying trust tiers. Let $\mathcal{P}_H$ and $\mathcal{P}_L$ denote the output distributions corresponding to high-trust ($\mathcal{T}=0$, minimum noise) and low-trust ($\mathcal{T}=1$, maximum noise) entities after the TADP-RME transformation. The Kullback–Leibler (KL) divergence between these distributions satisfies
\begin{equation}
D_{\text{KL}}(\mathcal{P}_L \| \mathcal{P}_H) 
\geq 
\frac{d}{2}\left(r^2 - 2\ln r - 1\right)
\label{eqno:r}
\end{equation}
where $r = \sigma_{\min}/\sigma_{\max} < 1$, with $\sigma_{\min} = \sigma_{\mathcal{T}=0}$ and $\sigma_{\max} = \sigma_{\mathcal{T}=1}$. This result characterizes statistical separation \cite{cover2006elements} between outputs corresponding to different trust levels. The bound follows from the divergence between Gaussian distributions with different variances \cite{cover2006elements}, and reflects how varying trust levels produce distinguishable output distributions. It does not directly imply privacy leakage, as differential privacy bounds worst-case adversarial inference. From a reliability standpoint, the KL divergence characterizes separation between system responses under different trust levels. Larger divergence implies clearer separation between operational regimes, which can be interpreted as controlled behavior under varying risk conditions rather than unintended exposure.

\begin{corollary}
As $r \to 0$, we have $
\lim_{r \to 0} D_{\text{KL}}(\mathcal{P}_L \| \mathcal{P}_H) \to \infty$
indicating that outputs corresponding to significantly different trust levels become increasingly distinguishable, while similar trust levels produce comparable representations. This property enables controlled utility differentiation within the proposed framework.
\end{corollary}

\subsection{Computational Security Analysis}

\subsubsection{Combinatorial Complexity of RME Inversion}

\begin{theorem}
Let $\mathcal{A}$ be any algorithm attempting to invert the RME transformation $\varphi^{-1}: \mathbb{R}^{2d} \to \mathbb{R}^d$ without knowledge of the correct coordinate pairing. Then the size of the search space for inversion grows at least as
\begin{equation}
T_{\min}(d) \geq \frac{(2d)!}{2^d d!} \cdot R^d
\end{equation}
where $R$ denotes the number of feasible solutions per coordinate pair induced by the nonlinear transformation.
\end{theorem}

\begin{proof}
The inversion process can be decomposed into two independent sources of combinatorial complexity. The RME transformation maps each input coordinate into two output components but does not preserve explicit pairing information. Recovering the original structure therefore requires enumerating all possible pairings of $2d$ coordinates into $d$ unordered pairs. The number of such pairings is $\frac{(2d)!}{2^d d!}$. For each candidate pair $(a_k, b_k)$, inversion requires solving a nonlinear trigonometric system. Due to periodicity, each pair admits multiple feasible solutions, bounded by $R$. Since these two sources are independent, the total search space scales as $\frac{(2d)!}{2^d d!} \cdot R^d$. Thus, any exhaustive inversion strategy must explore a search space of this order.
\end{proof}

\subsubsection{Resilience Against Partial Knowledge Attacks}
\begin{figure}[ht]
    \centering
    \includegraphics[width=1\linewidth]{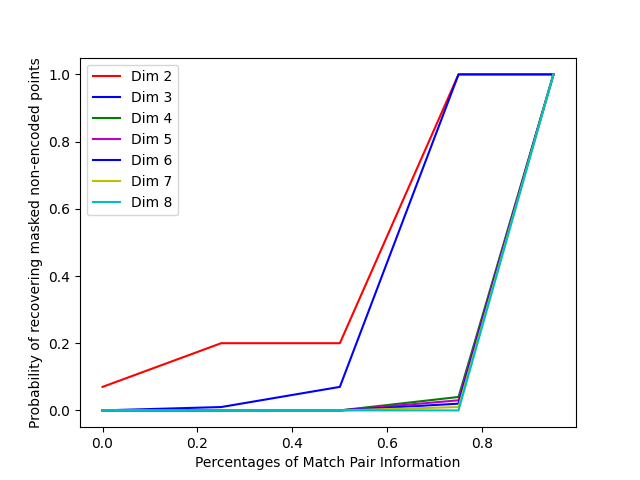}
    \caption{Recovery probability under partial pairing knowledge. While reconstruction remains feasible in low dimensions ($d=2,3$), success probability drops sharply as dimensionality increases, even with high levels of known pairings.}
    \label{fig:probgraph}
\end{figure}
\noindent We next consider an adversary with partial structural knowledge. Suppose an attacker correctly identifies $m$ coordinate pairings, leaving $l = d - m$ unknown pairs. The remaining search space is then $|\mathcal{S}_l| = \frac{(2l)!}{2^l l!} \cdot R^l$. This expression grows rapidly with $l$ due to combinatorial expansion effects. Therefore, unless a substantial fraction of pairings is known, the inversion problem remains computationally challenging. The RME transformation introduces a combinatorial barrier that is robust to partial information leakage, substantially increasing reconstruction difficulty. As shown in Fig.~\ref{fig:probgraph}, reconstruction remains achievable in low-dimensional settings ($d=2,3$), even with substantial prior knowledge (e.g., $75\%$ correct pairings). However, as the dimensionality increases, the probability of successful recovery declines rapidly. This behavior reflects the growth of the inversion space in RME, indicating that higher-dimensional embeddings significantly strengthen resistance against reconstruction attacks, even under partial knowledge. This analysis assumes an adversary without additional side information beyond partial pairing knowledge. More informed adversarial models may reduce the effective search space. It is important to note that while this dimensional expansion ($d \to 2d$) exponentially increases the combinatorial search space for an adversary attempting exact coordinate reconstruction, it does not destroy the utility for downstream machine learning tasks. Because the RME transformation applies deterministic, continuous trigonometric mappings, it empirically preserves class separability for downstream learning tasks. Consequently, lightweight downstream models (such as logistic regression) can still efficiently converge and achieve high classification accuracy without requiring an exponential increase in training data.
\subsection{Information-Theoretic Security}
\noindent We analyze the mechanism from an information-theoretic perspective to quantify how trust-adaptive noise reduces information leakage.
\subsubsection{Mutual Information Bounds}
\begin{theorem}
Let $X \in \mathbb{R}^d$ denote the original data and $Z_{\mathcal{T}} = \varphi(Y_{\mathcal{T}})$ denote the output of the TADP-RME mechanism, where $Y_{\mathcal{T}} = X + \mathcal{N}(0, \sigma_{\mathcal{T}}^2 I)$. Then the mutual information satisfies

\begin{equation}
I(X; Z_{\mathcal{T}}) \leq \frac{d}{2}\log\left(1 + \frac{\mathbb{E}[\|X\|^2]}{d\sigma_{\mathcal{T}}^2}\right)
\end{equation}
\end{theorem}

\begin{proof}
The mechanism can be decomposed into a Gaussian channel followed by a deterministic transformation. For the Gaussian channel $Y_{\mathcal{T}} = X + N$, standard results yield \cite{cover2006elements} $I(X; Y_{\mathcal{T}}) \leq \frac{d}{2}\log\left(1 + \frac{\mathbb{E}[\|X\|^2]}{d\sigma_{\mathcal{T}}^2}\right)$. Since $Z_{\mathcal{T}} = \varphi(Y_{\mathcal{T}})$ is deterministic, the data processing inequality implies \cite{cover2006elements} $I(X; Z_{\mathcal{T}}) \leq I(X; Y_{\mathcal{T}})$.
\end{proof}

\begin{corollary}As $\sigma_{\mathcal{T}}^2$ increases with $\mathcal{T}$, the mutual information decreases, indicating that $I(X; Z_{\mathcal{T}})$ decreases as $\sigma_{\mathcal{T}}^2$ increases, corresponding to reduced information leakage at higher protection levels.
\end{corollary}

\subsubsection{Geometric Distortion and Inversion Ambiguity}

\noindent The RME transformation further increases ambiguity in reconstruction.

\begin{proposition}
The mapping $\varphi(x_i)=(x_i \cos(\alpha x_i), x_i \sin(\alpha x_i))$ is non-injective and admits multiple valid inverse solutions.
\end{proposition}

\begin{proof}
Given $(a,b) = (x \cos(\alpha x), x \sin(\alpha x))$, we obtain $a^2 + b^2 = x^2$ implying $x = \pm \sqrt{a^2 + b^2}$. Additionally, the phase satisfies $\alpha x = \tan^{-1}(b/a) + 2\pi k, \quad k \in \mathbb{Z}$. Thus, multiple solutions exist, making inversion inherently ambiguous.
\end{proof}
{\textit{While differential privacy limits information leakage, the geometric transformation introduces structural ambiguity that increases resistance to reconstruction.}}

\section{Experiments}
\noindent We design a comprehensive experimental framework to evaluate the proposed method against a range of privacy-preserving mechanisms from a reliability perspective. Our evaluation focuses on three aspects: {\textit{($1$) quantification of the privacy–utility Pareto frontier, ($2$) the mechanism's structural resilience to inference attacks \cite{jayaraman2020evaluating}, and ($3$) the isolated empirical impact of Reverse Manifold Embedding via targeted ablation.}}

\subsection{Experimental Setup}
\subsubsection{Datasets and Preprocessing}
\noindent To demonstrate the scalability and generalizability of our framework, we conduct evaluations across three commonly adopted benchmarks of increasing complexity: MNIST, Fashion-MNIST, and CIFAR-10. This selection spans from simple grayscale digits to highly structured, high-dimensional natural images, providing evaluation across diverse data distributions and varying complexity levels. For consistency, all datasets are uniformly subsampled to $10,000$ training instances, normalized to the $[0, 1]$ interval, and flattened into one-dimensional feature vectors prior to any privacy transformations.

\subsubsection{Baselines}
\noindent We benchmark the proposed method against eight established privacy-preserving mechanisms, selected to represent current privacy mechanisms across three distinct paradigms. \cite{cummings2023dp}
\begin{itemize}
    \item {Noise-Injection DP Paradigms:} Standard Gaussian DP \cite{dong2022gdp} and Laplace DP \cite{dwork2014book}, representing classical fixed-budget mechanisms, alongside Personalized Differential Privacy (PDP) \cite{jorgensen2015pdp, ebadi2016pdp}, which dynamically scales privacy budgets per user.
    \item {Geometric and Projection Paradigms:} Random Projection Privacy \cite{bingham2001random, liu2019random} (dimensionality reduction coupled with additive noise) and Reconstruction-Resistant Privacy \cite{aggarwal2008privacy} (projection followed by noise injection and strict $L_2$ normalization).
    \item {Encoding and Hashing Paradigms:} Locality-Sensitive Hashing (LSH) Privacy \cite{indyk1998lsh} and Binary Encoding Privacy (incorporating probabilistic bit-flipping), which obscure data through discrete transformations.
    \item {Additive Noise Baseline:} A simple additive noise method is included as a control to isolate the value of formal DP scaling and geometric distortion.
\end{itemize}

\subsubsection{Evaluation Protocol and Reproducibility}
To ensure reproducibility and algorithmic transparency, the proposed framework and all baseline models are implemented in Python utilizing the Scikit-Learn and NumPy libraries. We use a controlled experimental environment where the trust score is evaluated across a discrete spectrum: $\tau \in \{0.0, 0.1, 0.25, 0.5, 0.75, 0.85, 0.95, 1.0\}$. For the underlying Gaussian mechanism, we rigorously bound the global sensitivity and clipping norm to $\Delta_2 = 1.0$, with a small failure probability of $\delta = 10^{-5}$. Consequently, the trust-adaptive privacy budget smoothly interpolates between a stringent privacy regime ($\epsilon_{\min} = 15.0$ at $\tau = 1.0$) and a high-utility regime ($\epsilon_{\max} = 80.0$ at $\tau = 0.0$). While these values exceed the range of typically considered in strict differential privacy settings, they reflect practical operating regimes where moderate privacy guarantees are acceptable. To reduce stochastic variance and improve statistical reliability, every experiment in our pipeline is averaged across five independent trials initialized with deterministic random seed offsets. For all utility and privacy metrics, we report the mean and standard deviation. Finally, to assess the statistical significance of the performance differential between TADP-RME and baseline methods, we employ paired $t$-tests, establishing statistical significance at the $p < 0.05$ threshold.

\subsection{Comprehensive Evaluation Metrics}
To evaluate the effectiveness of the proposed method from a reliability standpoint, we deploy a dual-faceted evaluation suite that quantifies both the retention of structural utility and the empirical resilience under adversarial conditions.
\subsubsection{Utility Preservation Metrics}

Traditional privacy evaluations often rely solely on downstream classification accuracy, which fails to capture structural degradation induced by protection mechanisms. We employ a multi-faceted utility assessment that quantifies both task-specific performance and structural preservation

\begin{itemize}
    \item {Linear Separability (Classification Utility):} We train a logistic regression classifier as a linear probe on protected representations, evaluating accuracy and weighted F1-score. The linear probe provides an estimate of the mechanism's ability to preserve class separability without relying on complex, parameterized models that may obscure underlying distortion. Let $\mathcal{M}$ be the privacy mechanism, $X$ the original data, and $\hat{X} = \mathcal{M}(X)$ the protected data. We define:
    \begin{equation}
        \text{Acc}_{\text{prot}} = \frac{1}{n} \sum_{i=1}^{n} \mathbf{1}[f_{\text{LR}}(\hat{x}_i) = y_i]
    \end{equation}
    where $f_{\text{LR}}$ is a logistic regression classifier trained on $\hat{X}$.

    \item {Topological Integrity ($k$-NN Overlap):} We quantify local structure preservation by measuring the overlap of $k$-nearest neighbor sets between original and protected feature spaces. For each sample $x_i$, let $\mathcal{N}_k(x_i)$ and $\mathcal{N}_k(\hat{x}_i)$ denote the $k$ nearest neighbors in the original and protected spaces, respectively. The overlap ratio is:
    \begin{equation}
        \text{Overlap}_k = \frac{1}{n} \sum_{i=1}^{n} \frac{|\mathcal{N}_k(x_i) \cap \mathcal{N}_k(\hat{x}_i)|}{k}
    \end{equation}
    We report $\text{Overlap}_k$ for $k \in \{5, 10, 20\}$, where higher values indicate better preservation of local structure.

    \item {Global Distance Preservation:} We measure the preservation of global distance structure using Spearman's rank correlation \cite{spearman1904} between pairwise Euclidean distances in the original and protected spaces. For all pairs $(i,j)$, let $d_{ij} = \|x_i - x_j\|_2$ and $\hat{d}_{ij} = \|\hat{x}_i - \hat{x}_j\|_2$. The Spearman correlation is:
    \begin{equation}
       \rho = 1 - \frac{6 \sum_{i<j} (r_{ij} - \hat{r}_{ij})^2}{m(m^2 - 1)}
    \end{equation}
    where $r_{ij}$ and $\hat{r}_{ij}$ are the ranks of $d_{ij}$ and $\hat{d}_{ij}$, and $m = \binom{n}{2}$ is the number of pairwise distances.$\rho \in [-1,1]$, where $\rho=1$ indicates perfect rank-order preservation and $\rho=0$ indicates no monotonic relationship between distances.
\end{itemize}
These three metrics provide complementary perspectives: classification utility assesses task-specific performance, $k$-NN overlap captures local structure fidelity, and distance correlation measures global geometry preservation. Together, they provide a comprehensive assessment of utility retention under transformation.

\subsubsection{Adversarial Privacy Metrics}
We evaluate empirical privacy through three attack models, which are interpreted as failure events. Each model produces a normalized privacy score $\text{Priv} \in [0,1]$, where higher values indicate stronger resistance to adversarial inference. From a reliability perspective, this score is directly interpretable as $R = 1 - P_{\text{attack}},$ where $P_{\text{attack}}$ denotes the success probability of the corresponding attack. Under this formulation, adversarial success represents a failure event, and the privacy score quantifies the probability of avoiding such failure.
\begin{itemize}
    \item {Membership Inference Attack (MIA):} A logistic regression classifier is trained to distinguish training samples from non-training samples. The privacy score is defined as $\text{Priv}_{\text{MIA}} = 1 - 2|\text{AUC} - 0.5|$, where $\text{AUC}$ denotes the attack performance \cite{shokri2017membership}.
      \item {Attribute Inference Attack (AIA):} A logistic regression model is used to infer sensitive attributes (e.g., class labels) from protected data. The privacy score is defined as $\text{Priv}_{\text{AIA}} = 1 - (\text{Acc} - \text{Baseline})/(1 - \text{Baseline})$, where $\text{Baseline} = 1/C$ for $C$ classes.
     \item {Reconstruction Attack:} Ridge regression is employed to recover original features from protected representations \cite{carlini2021extracting}. The privacy score is defined as $\text{Priv}_{\text{Recon}} = 1 - \frac{\|\hat{x} - x\|_2}{\|x\|_2}$, where $\epsilon = \|\hat{x} - x\|_2 / \|x\|_2$ denotes normalized reconstruction error.
\end{itemize}
The overall privacy score is computed as the mean of the three components 
$\text{Priv}_{\text{Overall}} = \frac{\text{Priv}_{\text{MIA}} + \text{Priv}_{\text{AIA}} + \text{Priv}_{\text{Recon}}}{3}.$
This composite measure provides a unified assessment of empirical privacy, which can also be interpreted as system reliability under adversarial conditions. The arithmetic mean is adopted for its interpretability and equal weighting of complementary threat models, following established evaluation practices \cite{jayaraman2019evaluating}. Each component is normalized to $[0,1]$, where values closer to one indicate minimal adversarial success and therefore higher reliability.

\subsection{Result and Discussion:}
\subsubsection{Privacy-Utility Trade-off}
\begin{table*}[!t]
\centering
\caption{Trust-adaptive privacy-utility trade-off across MNIST, Fashion-MNIST, and CIFAR-10. The inverse trust score $\tau \in [0,1]$ controls the privacy budget as $\varepsilon_{\tau} = 80 - \tau \cdot 65$, where $\tau=0$ corresponds to fully trusted (maximum utility, $\varepsilon=80$) and $\tau=1$ corresponds to fully untrusted (maximum privacy, $\varepsilon=15$). Results show mean $\pm$ standard deviation across 5 independent seeds.}
\label{tab:tradeoff_full}
\setlength{\tabcolsep}{1pt}
\resizebox{\textwidth}{!}{%
\begin{tabular}{ccccccccccc}
\hline
\multicolumn{2}{c}{} & \multicolumn{3}{c}{\textbf{MNIST}} & \multicolumn{3}{c}{\textbf{Fashion-MNIST}} & \multicolumn{3}{c}{\textbf{CIFAR-10}} \\
\cline{3-11}
$\tau$ & $\varepsilon$ & \textbf{Acc.} & \textbf{Priv.} & \textbf{Recon.} & \textbf{Acc.} & \textbf{Priv.} & \textbf{Recon.} & \textbf{Acc.} & \textbf{Priv.} & \textbf{Recon.} \\
\hline
0.0 & 80.0 & $90.0 \pm 0.7$ & $0.399 \pm 0.006$ & $0.114 \pm 0.000$ & $81.7 \pm 0.6$ & $0.429 \pm 0.001$ & $0.104 \pm 0.001$ & $38.4 \pm 1.4$ & $0.694 \pm 0.004$ & $0.296 \pm 0.004$ \\
0.1 & 73.5 & $72.7 \pm 0.4$ & $0.561 \pm 0.002$ & $0.517 \pm 0.001$ & $63.1 \pm 0.9$ & $0.585 \pm 0.003$ & $0.457 \pm 0.002$ & $15.8 \pm 2.3$ & $0.804 \pm 0.003$ & $0.479 \pm 0.001$ \\
0.25 & 63.8 & $68.7 \pm 0.7$ & $0.586 \pm 0.005$ & $0.551 \pm 0.001$ & $59.6 \pm 0.7$ & $0.600 \pm 0.006$ & $0.475 \pm 0.002$ & $14.8 \pm 1.4$ & $0.813 \pm 0.002$ & $0.499 \pm 0.003$ \\
0.5 & 47.5 & $57.8 \pm 1.0$ & $0.631 \pm 0.004$ & $0.616 \pm 0.002$ & $51.6 \pm 1.1$ & $0.640 \pm 0.005$ & $0.515 \pm 0.002$ & $12.4 \pm 1.4$ & $0.831 \pm 0.003$ & $0.539 \pm 0.002$ \\
0.75 & 31.3 & $40.6 \pm 1.5$ & $0.705 \pm 0.003$ & $0.700 \pm 0.002$ & $38.7 \pm 0.8$ & $0.704 \pm 0.004$ & $0.570 \pm 0.003$ & $11.2 \pm 1.0$ & $0.858 \pm 0.006$ & $0.600 \pm 0.001$ \\
0.85 & 24.8 & $32.2 \pm 1.1$ & $0.737 \pm 0.003$ & $0.738 \pm 0.002$ & $32.4 \pm 1.0$ & $0.729 \pm 0.005$ & $0.598 \pm 0.003$ & $10.9 \pm 1.1$ & $\mathbf{0.862 \pm 0.005}$ & $\mathbf{0.619 \pm 0.002}$ \\
0.95 & 18.3 & $24.5 \pm 0.4$ & $0.764 \pm 0.006$ & $0.767 \pm 0.002$ & $25.6 \pm 1.1$ & $0.762 \pm 0.004$ & $0.622 \pm 0.001$ & $10.4 \pm 1.0$ & $0.858 \pm 0.004$ & $0.601 \pm 0.002$ \\
1.0 & 15.0 & $26.6 \pm 1.2$ & $\mathbf{0.782 \pm 0.003}$ & $\mathbf{0.760 \pm 0.001}$ & $26.8 \pm 1.1$ & $\mathbf{0.756 \pm 0.006}$ & $\mathbf{0.618 \pm 0.001}$ & $10.1 \pm 0.8$ & $0.851 \pm 0.004$ & $0.576 \pm 0.003$ \\
\hline
\end{tabular}%
}
\end{table*}

\noindent Table~\ref{tab:tradeoff_full} summarizes the trust-adaptive privacy-utility trade-off across datasets. From a reliability perspective, the reported privacy scores can be interpreted as the probability of avoiding adversarial failure, where higher values indicate stronger resistance to inference attacks.
\begin{figure}[ht]
    \centering
    \includegraphics[width=\linewidth]{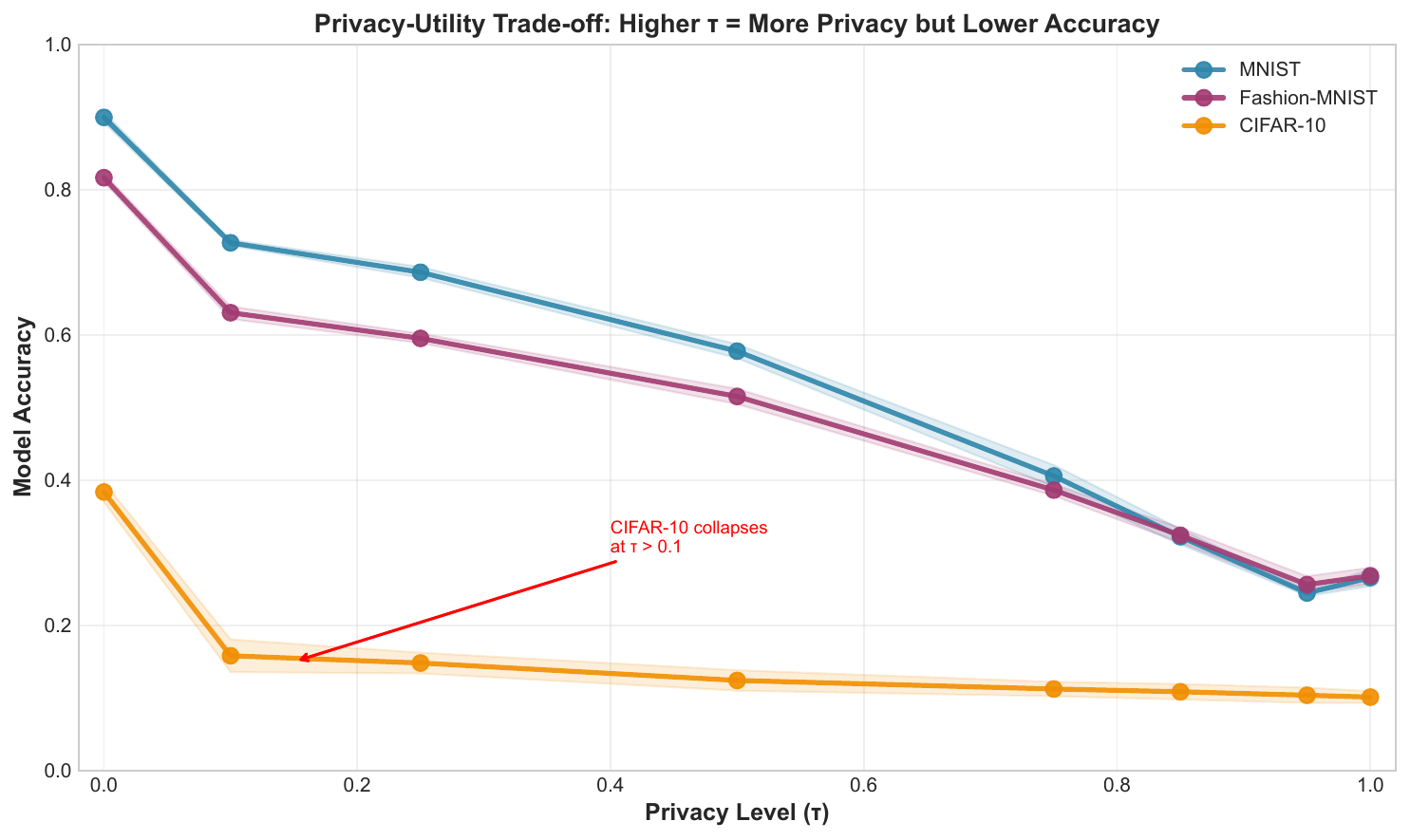}
    \caption{Privacy-utility trade-off across datasets. Accuracy decreases as $\tau$ increases, with CIFAR-10 exhibiting rapid degradation at low privacy levels.}
    \label{fig:tradeoff_curves}
\end{figure}

\begin{figure}[ht]
    \centering
    \includegraphics[width=0.9\linewidth]{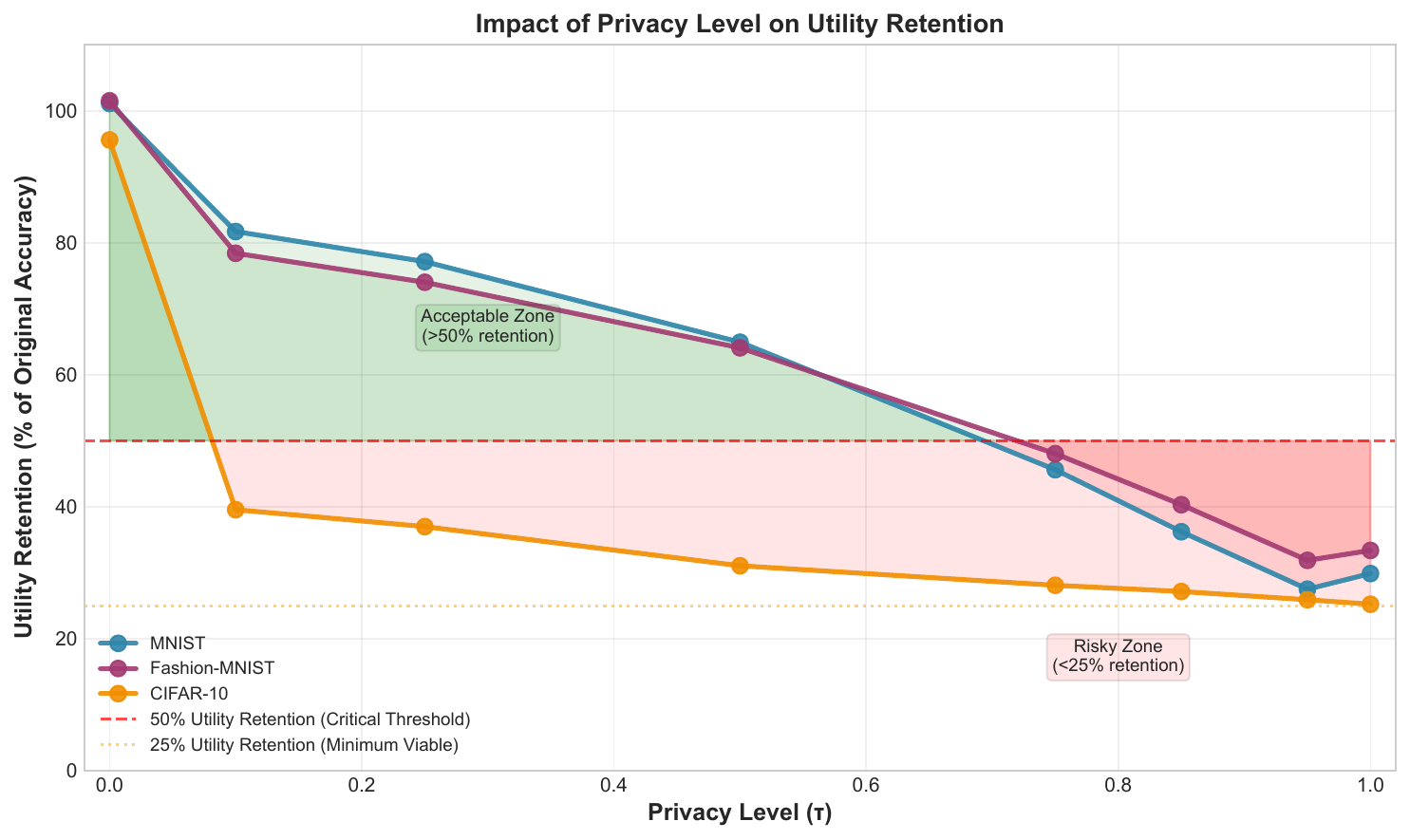}
    \caption{Utility retention across privacy levels $\tau$. Higher $\tau$ leads to reduced utility, with a critical transition around $\tau \approx 0.5$.}
    \label{fig:utility_retention}
\end{figure}

Figures~\ref{fig:tradeoff_curves} and~\ref{fig:utility_retention}, together with Table~\ref{tab:tradeoff_full}, show the privacy-utility trade-off observed for the proposed TADP-RME framework. As the inverse trust score $\tau$ increases, the privacy budget $\varepsilon_{\tau}$ decreases, resulting in stronger noise injection and higher empirical privacy scores, which correspond to increased reliability against adversarial inference, at the cost of reduced utility. Across all datasets, classification accuracy generally decreases while privacy scores increase, indicating improved resistance to adversarial failure. At $\tau=0$, corresponding to the fully trusted regime, all datasets achieve maximum utility with relatively low resistance to adversarial inference. In contrast, at $\tau=1$, the mechanism operates under higher privacy settings, where empirical privacy scores are highest, corresponding to stronger reliability against adversarial inference, but utility is reduced. A noticeable transition occurs around $\tau \approx 0.5$, where utility drops substantially, falling below approximately 60\% for MNIST and Fashion-MNIST and to lower values for CIFAR-10. This regime may represent a practical operating point balancing utility and reliability under adversarial conditions, consistent with the moderate privacy setting ($\varepsilon = 47.5$). Dataset-specific behavior highlights the role of data complexity. MNIST exhibits relatively stable performance across privacy levels, while Fashion-MNIST shows moderate sensitivity. In contrast, CIFAR-10 exhibits rapid utility degradation even at low $\tau$, suggesting that high-dimensional datasets may be more sensitive to perturbations. This behavior indicates that increasing privacy levels directly reduces adversarial success probability, thereby improving system reliability while introducing a trade-off with predictive performance. Figure~\ref{fig:utility_retention} further highlights two distinct operating regions:{\textit{(i) a higher utility regime above 50\% retention and (ii) a low utility regime below 25\%, where performance approaches random guessing.}} CIFAR-10 enters the low-utility regime at lower $\tau$ values, indicating its higher sensitivity to privacy constraints.
\begin{figure}[ht]
    \centering
    \includegraphics[width=\linewidth]{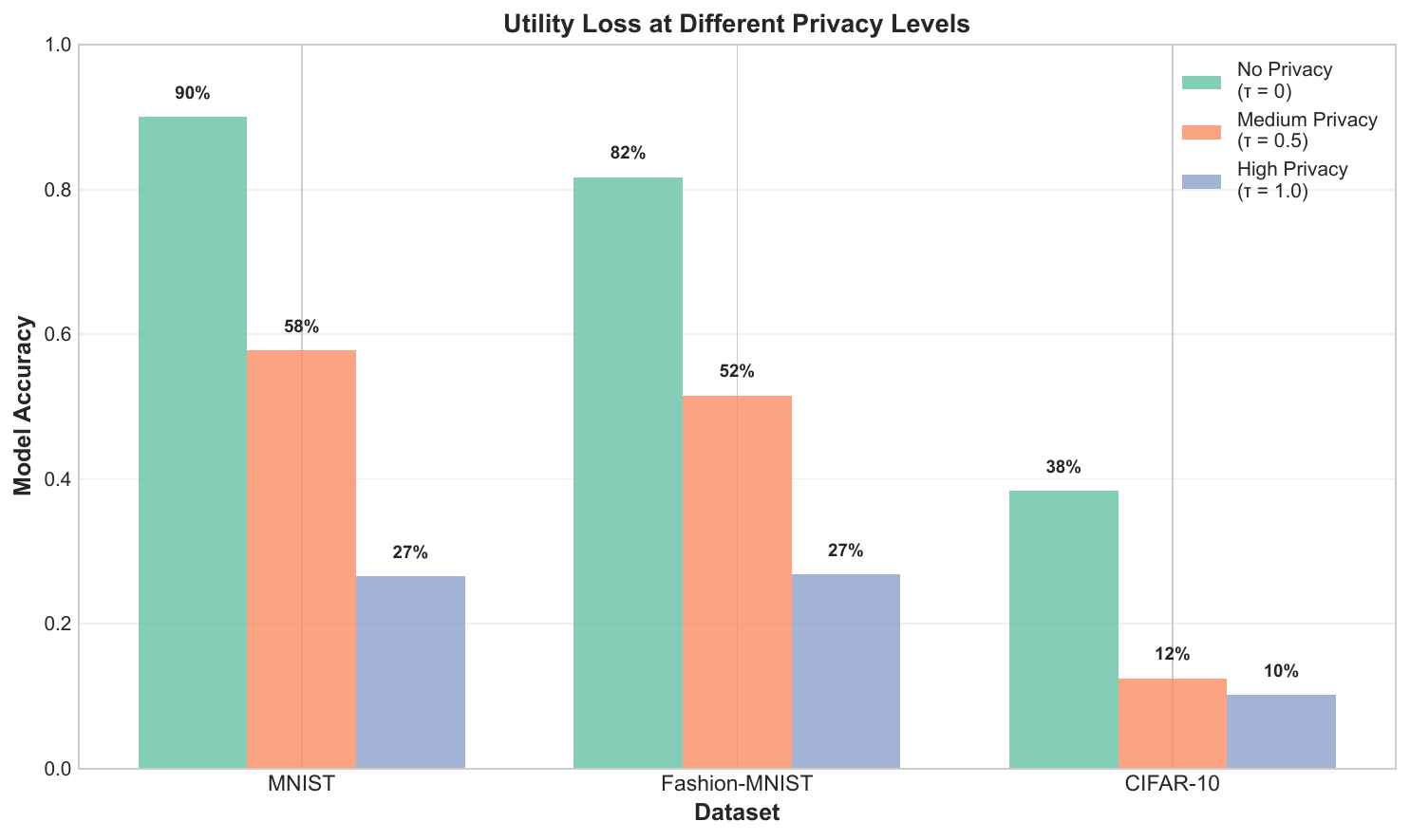}
    \caption{Utility comparison across privacy regimes ($\tau=0$, $0.5$, and $1$). Moderate privacy provides a balance between accuracy and protection, while strong privacy leads to significant degradation.}
    \label{fig:utility_loss}
\end{figure}
Figure~\ref{fig:utility_loss} provides a discrete comparison across representative privacy regimes. Moderate privacy ($\tau=0.5$) retains moderate utility across all datasets, while strong privacy ($\tau=1$) leads to significant degradation, particularly for CIFAR-10. These results suggest that moderate privacy provides a balance between utility and protection. Overall, the proposed framework provides a controllable mechanism for exploring the trade-off between utility and reliability under adversarial conditions, enabling flexible adjustment of privacy levels while maintaining usable performance.

\subsubsection{Comparison with Baseline Methods}
\begin{table*}[ht]
\centering
\caption{Comparison with baseline methods at matched privacy budgets. For DP baselines, we evaluate at $\varepsilon = 15$ (strong privacy), $\varepsilon = 47.5$ (moderate privacy), and $\varepsilon = 80$ (weak privacy). TADP-RME is evaluated at $\tau = 1$, $\tau = 0.5$, and $\tau = 0$, corresponding to $\varepsilon = 15$, $47.5$, and $80$, respectively. Non-DP baselines are evaluated with their default configurations. Results show mean accuracy and privacy score across 5 independent seeds, where privacy score is the average of membership inference, reconstruction, and attribute inference privacy scores (higher values indicate stronger empirical privacy).}
\label{tab:baseline_comparison}
\setlength{\tabcolsep}{5pt}
\begin{tabular}{lcccccc}
\hline
\multirow{2}{*}{\textbf{Method}} & \multicolumn{2}{c}{\textbf{MNIST}} & \multicolumn{2}{c}{\textbf{Fashion-MNIST}} & \multicolumn{2}{c}{\textbf{CIFAR-10}} \\
\cline{2-7}
 & \textbf{Accuracy} & \textbf{Privacy} & \textbf{Accuracy} & \textbf{Privacy} & \textbf{Accuracy} & \textbf{Privacy} \\
\hline

\multicolumn{7}{c}{\textbf{$\varepsilon = 15$ (Strong Privacy)}} \\
\hline
Gaussian DP & $26.8 \pm 1.2$ & $0.784 \pm 0.007$ & $27.1 \pm 1.5$ & $0.757 \pm 0.013$ & $11.0 \pm 0.9$ & $0.848 \pm 0.004$ \\
Laplace DP & $58.9 \pm 1.4$ & $0.585 \pm 0.010$ & $55.1 \pm 0.7$ & $0.625 \pm 0.007$ & $13.2 \pm 0.8$ & $0.804 \pm 0.003$ \\
Personalized DP & $12.1 \pm 0.7$ & $0.839 \pm 0.003$ & $12.9 \pm 0.8$ & $0.815 \pm 0.005$ & $10.3 \pm 0.6$ & $0.861 \pm 0.003$ \\
\textbf{TADP-RME ($\tau=1$)} & $\mathbf{26.6 \pm 1.2}$ & $\mathbf{0.782 \pm 0.003}$ & $\mathbf{26.8 \pm 1.1}$ & $\mathbf{0.756 \pm 0.006}$ & $\mathbf{10.1 \pm 0.8}$ & $\mathbf{0.851 \pm 0.004}$ \\

\hline
\multicolumn{7}{c}{\textbf{$\varepsilon = 47.5$ (Moderate Privacy)}} \\
\hline
Gaussian DP & $57.7 \pm 0.7$ & $0.605 \pm 0.001$ & $53.2 \pm 1.3$ & $0.634 \pm 0.007$ & $13.7 \pm 0.8$ & $0.818 \pm 0.003$ \\
Laplace DP & $80.9 \pm 0.8$ & $0.468 \pm 0.006$ & $70.7 \pm 1.1$ & $0.531 \pm 0.008$ & $22.2 \pm 0.9$ & $0.728 \pm 0.003$ \\
Personalized DP & $12.1 \pm 0.7$ & $0.840 \pm 0.006$ & $12.9 \pm 0.8$ & $0.814 \pm 0.005$ & $10.3 \pm 0.6$ & $0.854 \pm 0.003$ \\
\textbf{TADP-RME ($\tau=0.5$)} & $\mathbf{57.8 \pm 1.0}$ & $\mathbf{0.631 \pm 0.004}$ & $\mathbf{51.6 \pm 1.1}$ & $\mathbf{0.640 \pm 0.005}$ & $\mathbf{12.4 \pm 1.4}$ & $\mathbf{0.818 \pm 0.003}$ \\

\hline
\multicolumn{7}{c}{\textbf{$\varepsilon = 80$ (Weak Privacy)}} \\
\hline
Gaussian DP & $71.0 \pm 0.2$ & $0.530 \pm 0.004$ & $62.8 \pm 1.0$ & $0.581 \pm 0.004$ & $16.9 \pm 1.1$ & $0.776 \pm 0.004$ \\
Laplace DP & $84.3 \pm 1.1$ & $0.445 \pm 0.007$ & $73.7 \pm 0.5$ & $0.506 \pm 0.006$ & $27.5 \pm 1.7$ & $0.701 \pm 0.002$ \\
Personalized DP & $12.1 \pm 0.7$ & $0.841 \pm 0.004$ & $12.9 \pm 0.8$ & $0.816 \pm 0.006$ & $10.3 \pm 0.6$ & $0.858 \pm 0.003$ \\
\textbf{TADP-RME ($\tau=0$)} & $\mathbf{90.0 \pm 0.7}$ & $0.399 \pm 0.006$ & $\mathbf{81.7 \pm 0.6}$ & $0.429 \pm 0.001$ & $\mathbf{38.4 \pm 1.4}$ & $0.694 \pm 0.004$ \\

\hline
\multicolumn{7}{c}{\textbf{Non-DP Baselines}} \\
\hline
Random Projection & $11.5 \pm 3.7$ & $0.179 \pm 0.002$ & $15.2 \pm 5.7$ & $0.171 \pm 0.003$ & $9.9 \pm 1.2$ & $0.331 \pm 0.003$ \\
Additive Noise & $85.9 \pm 0.5$ & $0.425 \pm 0.004$ & $76.8 \pm 0.6$ & $0.458 \pm 0.005$ & $39.6 \pm 1.6$ & $0.638 \pm 0.005$ \\
LSH & $9.3 \pm 1.4$ & $0.232 \pm 0.004$ & $11.7 \pm 4.2$ & $0.212 \pm 0.004$ & $9.7 \pm 0.6$ & $0.379 \pm 0.004$ \\
Binary Encoding & $11.6 \pm 0.7$ & $0.841 \pm 0.003$ & $10.0 \pm 0.7$ & $0.826 \pm 0.001$ & $10.4 \pm 0.8$ & $0.792 \pm 0.003$ \\
Reconstruction-Resistant & $11.1 \pm 3.0$ & $0.219 \pm 0.002$ & $11.2 \pm 3.6$ & $0.201 \pm 0.003$ & $9.6 \pm 0.8$ & $0.357 \pm 0.004$ \\

\hline
\end{tabular}
\end{table*}
\noindent Table~\ref{tab:baseline_comparison} presents a comparison between the proposed framework and representative privacy-preserving methods under matched privacy budgets. From a reliability perspective, the reported privacy scores can be interpreted as resistance to adversarial failure. For differential privacy baselines, we consider three regimes corresponding to strong ($\varepsilon=15$), moderate ($\varepsilon=47.5$), and weak ($\varepsilon=80$) privacy. At strong privacy ($\varepsilon=15$), all methods exhibit reduced utility due to increased noise levels. Personalized DP achieves the highest privacy scores but at the cost of near-random accuracy across all datasets. In contrast, the proposed method (TADP-RME) achieves comparable privacy scores while slightly higher utility, suggesting a more favorable balance between reliability and usability under adversarial conditions. At moderate privacy ($\varepsilon=47.5$), classical mechanisms such as Gaussian and Laplace DP achieve higher accuracy but exhibit noticeably lower privacy scores. The proposed method, evaluated at the corresponding operating point ($\tau=0.5$ in Table~\ref{tab:tradeoff_full}), achieves higher privacy scores, indicating improved resistance to adversarial inference, while maintaining comparable accuracy. This indicates a more favorable trade-off between utility and reliability compared to standard noise-based approaches. At weak privacy ($\varepsilon=80$), utility improves for all methods, particularly Laplace DP and additive noise. However, these gains come at the cost of reduced privacy, highlighting the trade-off in fixed-noise mechanisms that cannot simultaneously preserve high utility and strong resistance to adversarial inference. Among non-DP baselines, Random Projection and LSH exhibit poor utility, indicating that aggressive structural transformations may degrade task performance. Binary encoding achieves high privacy but results in near-random accuracy, limiting its applicability for downstream tasks. Additive noise achieves high accuracy but provides relatively limited privacy protection. In this context, higher privacy scores correspond to lower adversarial success probability, and therefore reflect improved system reliability under inference attacks. {\textit{By combining trust-adaptive noise with geometric transformation, TADP-RME improves empirical privacy while limiting utility degradation, provides a competitive balance between utility and reliability compared to both noise-based and transformation-based baselines.}}To validate whether the observed differences are statistically significant, we perform paired $t$-tests between TADP-RME and each baseline across five independent runs. The results indicate that, in most cases, the improvements in privacy scores achieved by TADP-RME, corresponding to increased reliability, at matched privacy budgets are statistically significant ($p < 0.05$), while differences in accuracy are generally comparable or exhibit smaller variance. These findings support that the observed privacy-utility trade-offs are not due to random variation, but reflect consistent performance trends across datasets.

\subsubsection{Attack Resilience Analysis}
\begin{table*}[ht]
\centering
\caption{Attack resilience at $\tau=0.0$ (weak privacy, $\varepsilon=80$) and $\tau=1$ (maximum privacy, $\varepsilon=15$). MIA (membership inference attack) privacy score measures protection against training set membership leakage, where $1$ indicates random guessing. Reconstruction privacy score quantifies resistance to feature inversion, with higher values indicating greater reconstruction difficulty. AIA (attribute inference attack) privacy score measures protection against label inference, normalized against random guessing baseline ($1/C$ for $C$ classes). All scores are normalized to $[0,1]$, where $1$ represents maximum resistance to adversarial inference. Results show mean $\pm$ standard deviation across 5 seeds.}
\label{tab:attack_summary}
\begin{tabular}{lcccccc}
\hline
\multirow{2}{*}{\textbf{Dataset}} & \multicolumn{3}{c}{\textbf{$\tau=0.0$}} & \multicolumn{3}{c}{\textbf{$\tau=1$}} \\
\cline{2-7}
 & \textbf{MIA} & \textbf{Recon} & \textbf{AIA} & \textbf{MIA} & \textbf{Recon} & \textbf{AIA} \\
\hline
MNIST & $0.541 \pm 0.006$ & $0.114 \pm 0.000$ & $0.543 \pm 0.006$ & $0.988 \pm 0.004$ & $0.532 \pm 0.000$ & $0.827 \pm 0.006$ \\
Fashion-MNIST & $0.592 \pm 0.005$ & $0.104 \pm 0.001$ & $0.591 \pm 0.005$ & $0.984 \pm 0.011$ & $0.461 \pm 0.001$ & $0.822 \pm 0.008$ \\
CIFAR-10 & $0.895 \pm 0.004$ & $0.296 \pm 0.004$ & $0.891 \pm 0.004$ & $0.987 \pm 0.009$ & $0.576 \pm 0.003$ & $0.981 \pm 0.006$ \\
\hline
\end{tabular}
\end{table*}
\noindent Table~\ref{tab:attack_summary} evaluates the empirical privacy of the proposed framework using three complementary attack models: membership inference (MIA), attribute inference (AIA), and reconstruction attacks. Each metric is normalized to $[0,1]$, where higher values indicate stronger privacy, and the overall score represents their arithmetic mean. From a reliability perspective, these scores correspond to resistance against adversarial failure, where higher values indicate improved system reliability. At low privacy ($\tau=0$), the privacy scores for all three attack models are relatively lower, suggesting that the protected representations still retain exploitable information. In particular, reconstruction privacy is weakest in this regime, reflecting the ability of an adversary to recover original features with low normalized error. As the privacy level increases, all three metrics show an increasing trend, indicating reduced adversarial success probability and improved reliability. The MIA score approaches values close to $1$, indicating that the attack classifier performs no better than random guessing (AUC $\approx 0.5$), thus reducing membership leakage. Similarly, AIA scores increase toward the random baseline ($1/C$), showing reduced predictability of sensitive attributes from protected embeddings. Reconstruction privacy exhibits the most pronounced improvement with increasing $\tau$. This trend signifies that higher perturbation levels increase reconstruction error and reduce the feasibility of inversion attacks. inversion attacks. This suggests that the combination of noise injection and geometric transformation introduces structural distortion that reduces feature-level recoverability. Dataset-specific trends provide additional insight into the behavior of the method. CIFAR-10 exhibits relatively higher baseline privacy due to its inherent complexity, but still shows consistent improvement across all attack metrics. In contrast, MNIST and Fashion-MNIST demonstrate more prominent relative improvements, indicating that privacy mechanisms more strongly affect exploitable structure in simpler datasets. Overall, the consistent improvement across MIA, AIA, and reconstruction metrics suggests that the proposed TADP-RME provides comprehensive protection against diverse inference attacks, thereby improving system reliability under adversarial conditions. {\textit{The alignment between the three components of the composite privacy score further indicates that the proposed method performs consistently across different attack models, ensuring stable reliability across multiple adversarial scenarios, without relying on a single threat scenario.}}

\subsubsection{Structural Preservation Analysis}
\begin{figure*}[ht]
    \centering
    \includegraphics[width=\textwidth]{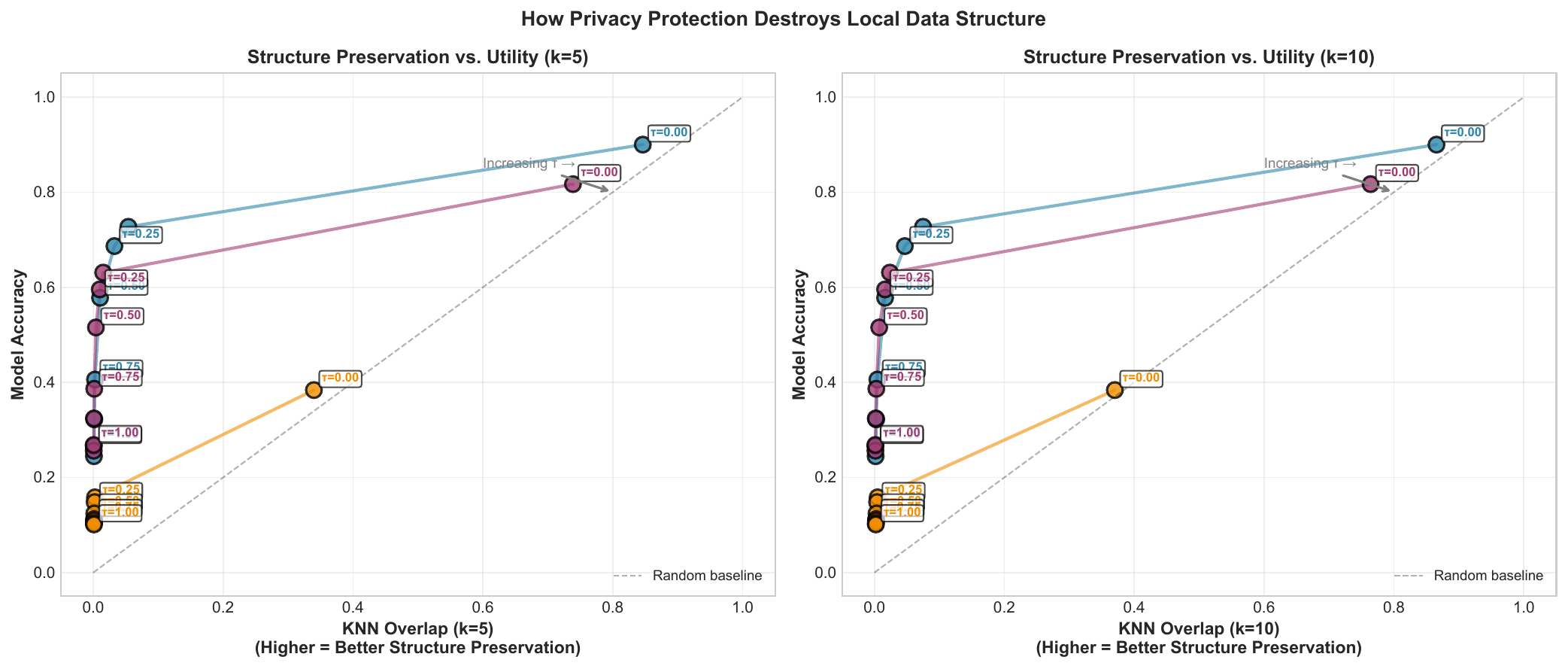}
    \caption{Structure preservation versus utility under varying privacy levels $\tau$, measured using k-NN overlap for $k=5$ and $k=10$. Higher overlap indicates better preservation of local geometry. Results are shown for MNIST (blue), Fashion-MNIST (purple), and CIFAR-10 (orange). As $\tau$ increases, both structure preservation and accuracy decrease. The trends are consistent across both neighborhood sizes, with $k=10$ exhibiting slightly smoother degradation compared to the more locally sensitive $k=5$, indicating that the observed behavior is consistent across different neighborhood scales.}
    \label{fig:knn_structure}
\end{figure*}
\noindent Figure~\ref{fig:knn_structure} analyzes the relationship between structural preservation and model utility under varying privacy levels $\tau$, using k-NN overlap as a measure of local geometric consistency, which reflects the preservation of structural information under adversarial conditions. At low privacy levels ($\tau \approx 0$), high k-NN overlap is observed across all datasets, signifying that neighborhood relationships are largely preserved. This is associated with higher classification accuracy, as the underlying data structure remains intact. As $\tau$ increases, both k-NN overlap and accuracy decrease, indicating increasing structural disruption and reduced exploitable information, reflecting progressive distortion of local neighborhoods due to noise injection and geometric transformation. The degradation trend is similar across both neighborhood sizes ($k=5$ and $k=10$), indicating that the trend is not sensitive to the choice of $k$. However, $k=5$ exhibits slightly sharper declines, as it captures more localized relationships, while $k=10$ shows smoother degradation due to its broader neighborhood definition. Despite these differences, both settings exhibit similar trends: increasing privacy systematically disrupts local data geometry, reducing the structural cues that can be exploited by adversarial inference. This effect is particularly pronounced for CIFAR-10, where k-NN overlap decreases significantly, which is aligned with classification accuracy. In contrast, MNIST and Fashion-MNIST retain higher structural consistency at moderate privacy levels, which explains their relatively stable utility. Overall, the results signify that the proposed TADP-RME achieves privacy not only through noise injection but also by disrupting local geometric structure. Since many inference and reconstruction attacks rely on preserving neighborhood relationships, {\textit{this structural degradation contributes to reduced attack effectiveness and improved resistance to adversarial inference, complementing formal differential privacy guarantees.}}This observation is consistent with the improvements in empirical privacy observed in Section~\ref{tab:attack_summary}, and further supports the interpretation of increased privacy as improved system reliability under adversarial conditions.

\subsubsection{Ablation Study}
\begin{table*}[ht]
\centering
\caption{Ablation study at $\tau=0.5$ (moderate privacy, $\varepsilon=47.5$). Noise-only applies the Gaussian mechanism without geometric transformation; embedding-only applies polar transformation without stochastic noise; fixed-$\tau$ (non-adaptive) uses the full pipeline with a constant trust level; full pipeline corresponds to the adaptive TADP-RME framework. Results are reported as mean $\pm$ standard deviation across 5 runs.}
\label{tab:ablation}
\setlength{\tabcolsep}{5pt}
\begin{tabular}{lcccccc}
\hline
\multirow{2}{*}{\textbf{Component}} & \multicolumn{2}{c}{\textbf{MNIST}} & \multicolumn{2}{c}{\textbf{Fashion-MNIST}} & \multicolumn{2}{c}{\textbf{CIFAR-10}} \\
\cline{2-7}
 & \textbf{Acc (\%)} & \textbf{Priv} & \textbf{Acc (\%)} & \textbf{Priv} & \textbf{Acc (\%)} & \textbf{Priv} \\
\hline
Noise Only & $58.2 \pm 0.6$ & $0.605 \pm 0.002$ & $54.7 \pm 0.6$ & $0.609 \pm 0.003$ & $14.7 \pm 0.6$ & $0.807 \pm 0.003$ \\
Embedding Only & $90.0 \pm 0.4$ & $0.400 \pm 0.004$ & $81.7 \pm 0.6$ & $0.429 \pm 0.001$ & $38.4 \pm 1.4$ & $0.694 \pm 0.004$ \\
Fixed-$\tau$ (Non-adaptive) & $58.0 \pm 1.0$ & $0.633 \pm 0.003$ & $51.9 \pm 0.9$ & $0.643 \pm 0.003$ & $12.6 \pm 0.5$ & $0.822 \pm 0.004$ \\
\hline
\textbf{Full Pipeline (Adaptive)} & $\mathbf{57.8 \pm 1.0}$ & $\mathbf{0.631 \pm 0.004}$ & $\mathbf{51.6 \pm 1.1}$ & $\mathbf{0.640 \pm 0.005}$ & $\mathbf{12.4 \pm 1.4}$ & $\mathbf{0.818 \pm 0.003}$ \\
\hline
\multicolumn{7}{c}{\textit{Privacy Improvement over Noise Only}} \\
\multicolumn{2}{c}{} & \textbf{+0.026} & \multicolumn{2}{c}{\textbf{+0.031}} & \multicolumn{2}{c}{\textbf{+0.011}} \\
\hline
\end{tabular}
\end{table*}
\noindent Table~\ref{tab:ablation} evaluates the contribution of individual components in the proposed framework from a reliability perspective at a fixed privacy level ($\tau=0.5$, $\varepsilon=47.5$). The noise-only variant applies Gaussian noise without geometric transformation. While it provides moderate privacy, its resistance to adversarial inference is limited, indicating that noise injection alone is insufficient to fully mitigate inference attacks. The embedding-only variant applies geometric transformation without stochastic noise. This configuration achieves higher accuracy due to the absence of noise-based perturbation, but provides weaker resistance to adversarial inference, as structural information remains partially exploitable by adversaries. The fixed-$\tau$ pipeline combines noise and embedding but operates with a constant trust level. Compared to noise-only, it achieves improved privacy, highlighting the benefit of incorporating geometric distortion to reduce exploitable structure. However, it lacks the flexibility of adaptive trust control. The full pipeline (TADP-RME) integrates both components within a unified framework. It consistently achieves higher privacy than the noise-only variant, corresponding to improved resistance to adversarial failure, while maintaining comparable accuracy. In particular, it improves privacy by $2.6\%$, $3.1\%$, and $1.1\%$ on MNIST, Fashion-MNIST, and CIFAR-10, respectively, without introducing additional degradation in utility. These results indicate {\textit{that noise injection and geometric transformation provide complementary benefits in improving resistance to adversarial inference.}} Noise contributes formal differential privacy guarantees, while embedding disrupts local structural patterns that noise alone cannot effectively conceal. Their combination is therefore essential for achieving a balanced trade-off between utility and reliability under adversarial conditions. These results indicate that combining stochastic noise with structural transformation reduces adversarial success probability, thereby improving system reliability while preserving practical utility.

\subsubsection{Parameter Sensitivity}
\begin{table}[ht]
\centering
\setlength{\tabcolsep}{1.2pt}   
\caption{Parameter sensitivity analysis on MNIST at $\tau=0.5$ ($\varepsilon=47.5$). Results are reported as mean $\pm$ standard deviation across 5 runs, where higher privacy scores indicate stronger resistance to adversarial inference.}
\label{tab:sensitivity}
\begin{tabular}{lcccc}
\hline
\textbf{Parameter} & \textbf{Value} & \textbf{Accuracy (\%)} & \textbf{Privacy Score} & \textbf{Reconstruction Error} \\
\hline
\multirow{4}{*}{$\varepsilon_{\min}$} & 10 & $55.0 \pm 1.1$ & $0.637 \pm 0.005$ & $0.628 \pm 0.002$ \\
 & 15 & $57.8 \pm 1.0$ & $0.631 \pm 0.004$ & $0.616 \pm 0.002$ \\
 & 20 & $59.8 \pm 0.7$ & $0.627 \pm 0.003$ & $0.607 \pm 0.002$ \\
 & 30 & $63.9 \pm 0.6$ & $0.606 \pm 0.004$ & $0.586 \pm 0.002$ \\
\hline
\multirow{4}{*}{$\varepsilon_{\max}$} & 40 & $45.8 \pm 1.0$ & $0.672 \pm 0.004$ & $0.666 \pm 0.003$ \\
 & 60 & $54.2 \pm 0.9$ & $0.653 \pm 0.006$ & $0.614 \pm 0.002$ \\
 & 80 & $57.8 \pm 1.0$ & $0.631 \pm 0.004$ & $0.616 \pm 0.002$ \\
 & 100 & $63.7 \pm 0.6$ & $0.609 \pm 0.004$ & $0.575 \pm 0.001$ \\
\hline
\multirow{3}{*}{Clip Norm $C$} & 0.5 & $31.4 \pm 1.3$ & $0.752 \pm 0.004$ & $0.747 \pm 0.001$ \\
 & 1.0 & $57.8 \pm 1.0$ & $0.631 \pm 0.004$ & $0.616 \pm 0.002$ \\
 & 2.0 & $78.7 \pm 0.7$ & $0.539 \pm 0.003$ & $0.454 \pm 0.001$ \\
\hline
\end{tabular}
\end{table}
\noindent Table~\ref{tab:sensitivity} analyzes the effect of key hyperparameters on the performance of TADP-RME from a reliability perspective at $\tau=0.5$. The minimum privacy budget $\varepsilon_{\min}$ controls the strength of protection in high-privacy regions. As $\varepsilon_{\min}$ increases from 10 to 30, classification accuracy improves steadily (from $55.0\%$ to $63.9\%$), while the privacy score decreases slightly. This indicates that relaxing the lower bound of privacy allows more information to be preserved, improving utility at the cost of reduced resistance to adversarial inference. A similar trend is observed for $\varepsilon_{\max}$, which determines the upper bound of the privacy budget. Increasing $\varepsilon_{\max}$ from 40 to 100 leads to a substantial gain in accuracy (from $45.8\%$ to $63.7\%$), accompanied by a decrease in privacy score, indicating reduced resistance to adversarial failure. This demonstrates that a larger upper bound enables higher utility in low-privacy regions, effectively increasing the utility ceiling of the framework. The clipping norm $C$ has a more pronounced impact on the trade-off. A smaller value ($C=0.5$) enforces strong regularization, resulting in high privacy ($0.752$), corresponding to strong resistance to adversarial inference, but significantly reduced accuracy ($31.4\%$). Conversely, a larger value ($C=2.0$) preserves more information, achieving high accuracy ($78.7\%$) but weaker resistance to adversarial inference ($0.539$). The intermediate setting ($C=1.0$) provides a balanced trade-off, maintaining reasonable accuracy ($57.8\%$) while preserving moderate privacy ($0.631$). Reconstruction error follows a consistent trend with privacy, decreasing as accuracy increases. This indicates that higher utility corresponds to improved reconstructability of the data, reinforcing the inherent trade-off between resistance to adversarial inference and information retention. Overall, the results demonstrate that our method offers intuitive and flexible control over the trade-off between utility and reliability under adversarial conditions through parameter tuning. By adjusting $\varepsilon_{\min}$, $\varepsilon_{\max}$, and $C$, practitioners can adapt the framework to different application requirements while maintaining predictable behavior. These trends indicate that parameter choices directly influence adversarial success probability, allowing controlled adjustment of system reliability alongside predictive performance.

\subsubsection{Global Structure and Efficiency Analysis}
\begin{figure}[ht]
    \centering
    \includegraphics[width=\linewidth]{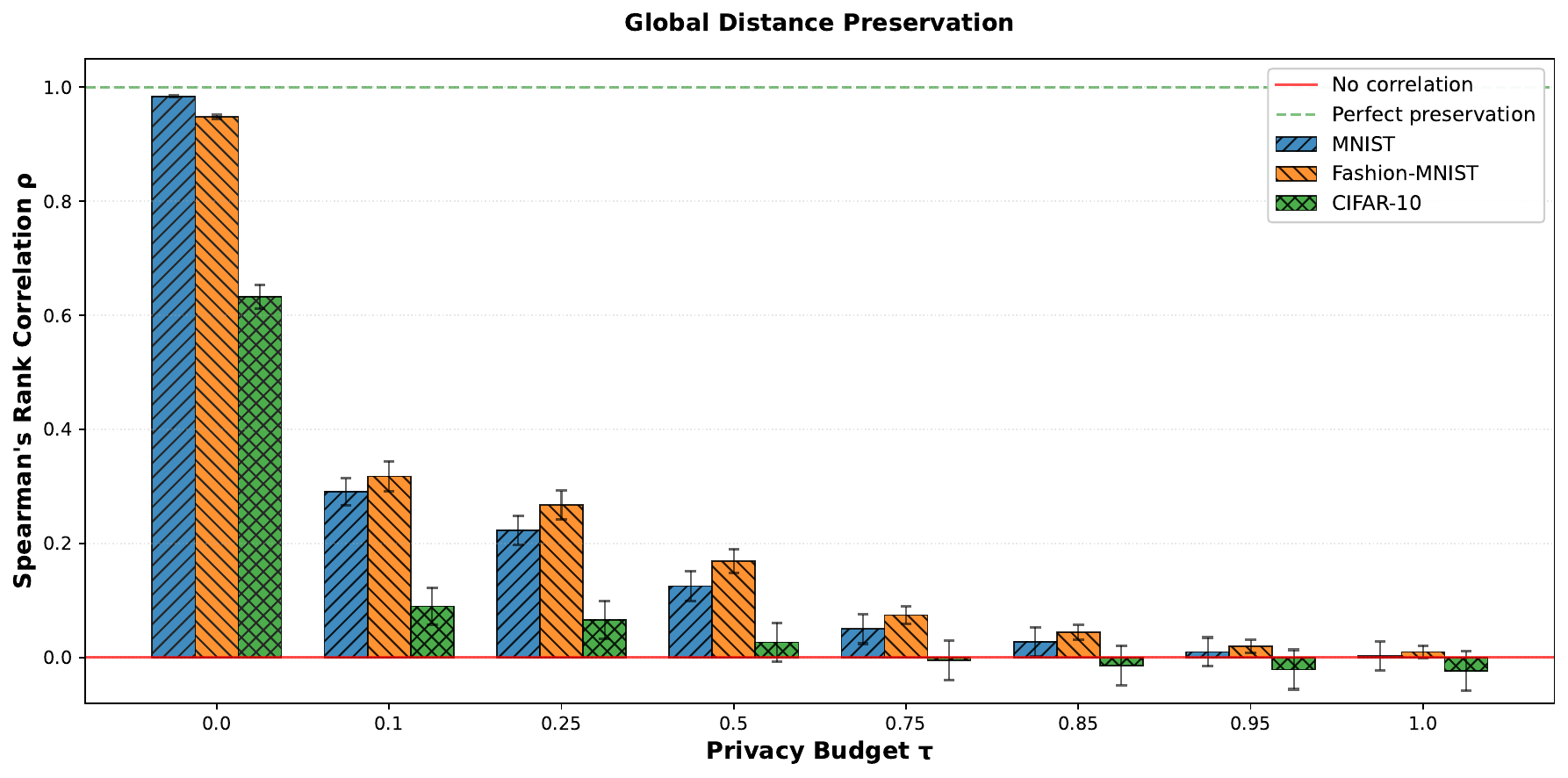}
    \caption{Global distance preservation measured using Spearman's rank correlation $\rho$. Increasing $\tau$ leads to rapid degradation of global geometry, with correlations approaching zero under strong privacy.}
    \label{fig:global_distance}
\end{figure}

\begin{figure}[ht]
    \centering
    \includegraphics[width=\linewidth]{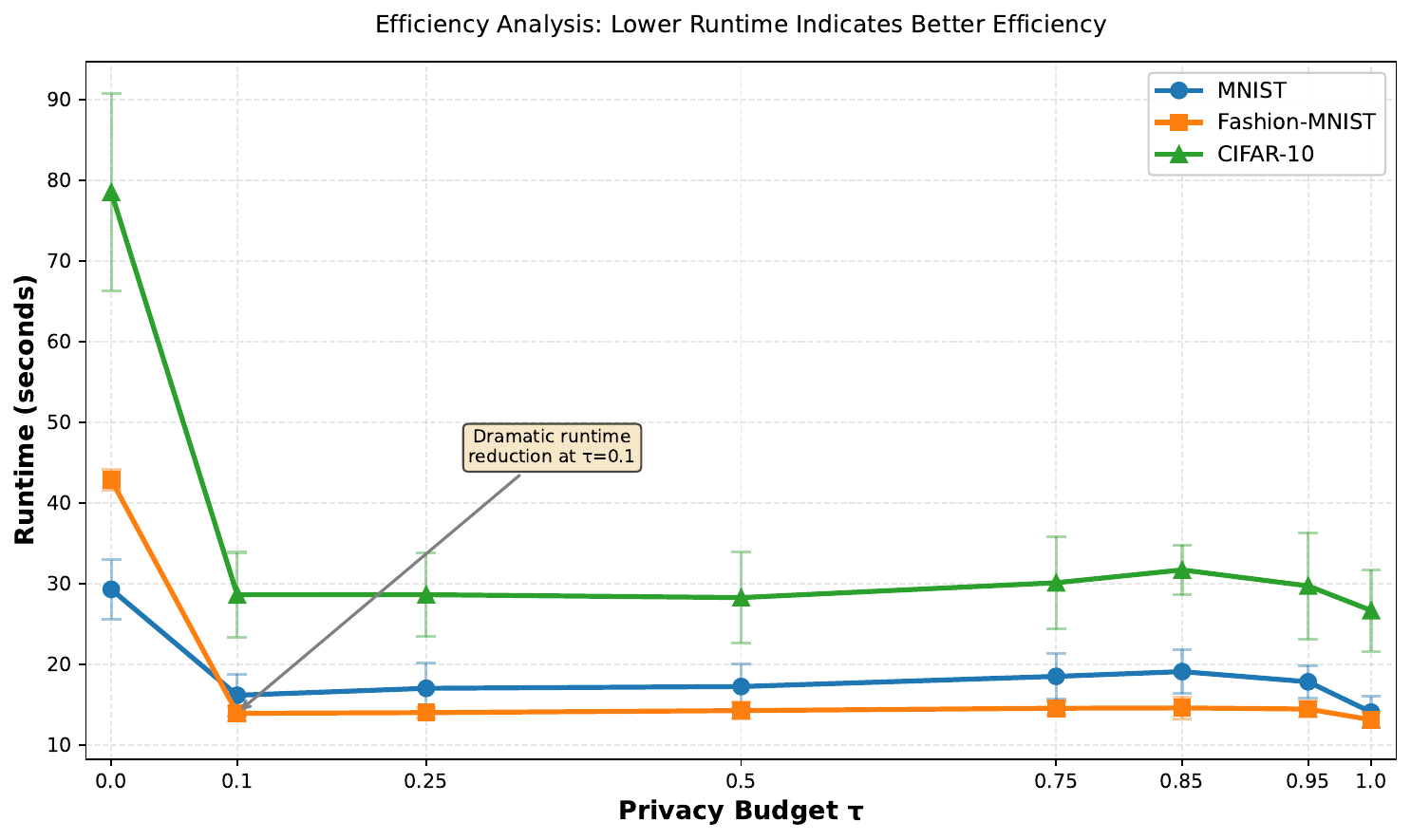}
    \caption{Computational overhead across privacy levels. Runtime decreases significantly as $\tau$ increases, indicating improved efficiency under stronger privacy settings.}
    \label{fig:runtime}
\end{figure}
\noindent Figures~\ref{fig:global_distance} and~\ref{fig:runtime} analyze global structure preservation and computational efficiency. Global distance preservation, measured using Spearman correlation, decreases sharply as $\tau$ increases, indicating significant disruption of global geometry and reduced availability of exploitable structural information. While MNIST and Fashion-MNIST maintain high correlation at $\tau=0$, all datasets approach near-zero correlation under strong privacy, which limits the effectiveness of structure-based inference attacks, with CIFAR-10 exhibiting the most severe degradation. In contrast, computational overhead decreases with increasing $\tau$. Runtime drops significantly from low to moderate privacy levels and remains stable thereafter, suggesting that stronger privacy reduces structural complexity of the transformed data and improves computational efficiency. These results demonstrate that the proposed TADP-RME {\textit{not only enhances resistance to adversarial inference by disrupting both local and global structures, but also improves computational efficiency at higher privacy levels.}} These observations indicate that structural disruption reduces adversarial success probability, thereby improving system reliability while simultaneously lowering computational overhead.

\subsubsection{Pareto Analysis}
\begin{figure}[ht]
    \centering
    \includegraphics[width=\linewidth]{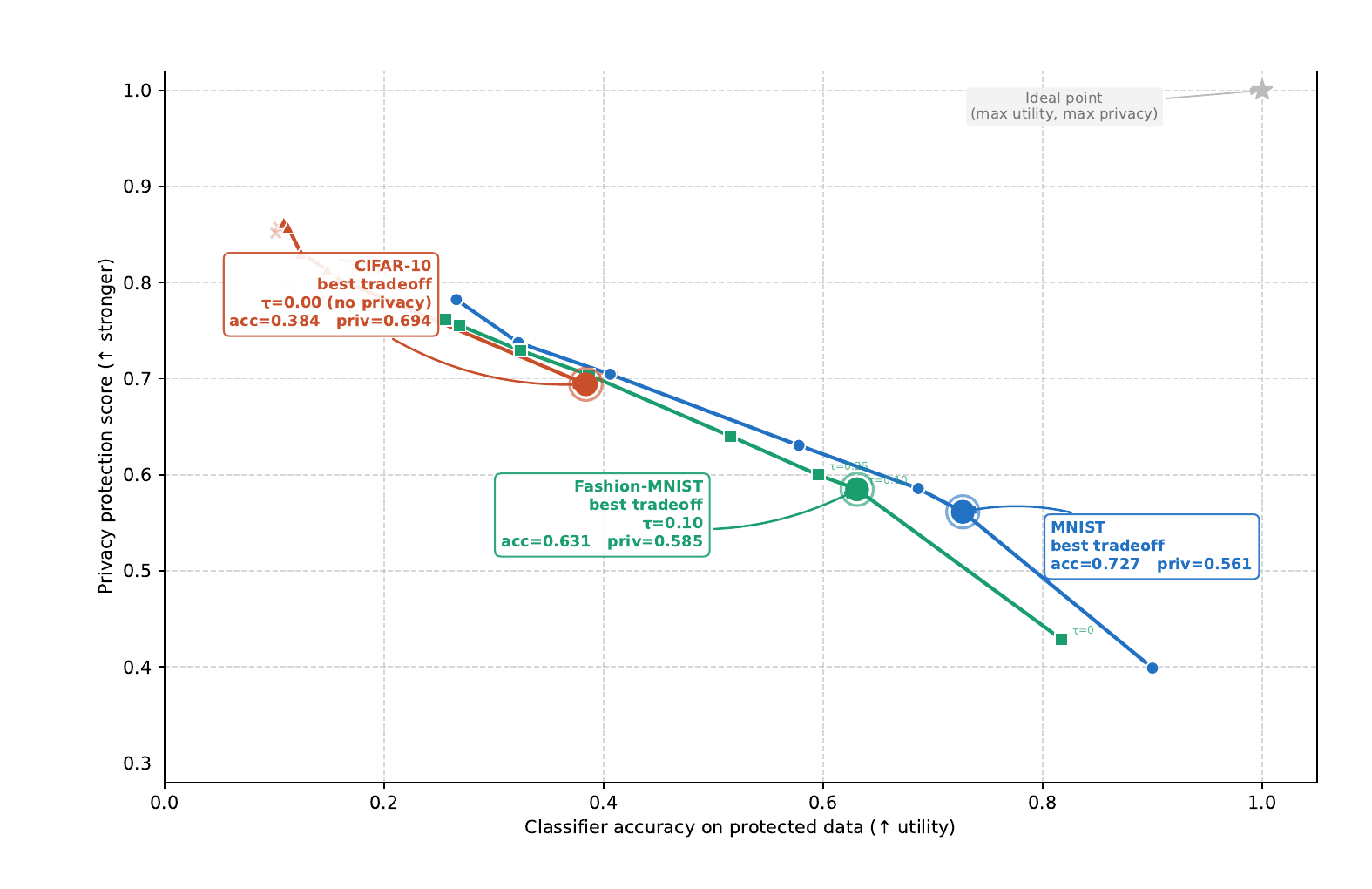}
    \caption{Privacy-utility Pareto frontier showing the trade-off between classification accuracy and empirical privacy, which can be interpreted as system reliability under adversarial conditions. Each point corresponds to a different trust level $\tau$.}
    \label{fig:pareto}
\end{figure}
\noindent Figure~\ref{fig:pareto} illustrates the privacy--utility Pareto frontier achieved by TADP-RME, reflecting the trade-off between utility and reliability under adversarial conditions across different trust levels $\tau$. For MNIST and Fashion-MNIST, an optimal trade-off is observed at moderate privacy levels ($\tau \approx 0.1$–$0.25$), where both accuracy and resistance to adversarial inference remain relatively high. In contrast, CIFAR-10 shows a more constrained frontier, where improvements in privacy (i.e., reduced adversarial success probability) lead to rapid degradation in accuracy. These results highlight the flexibility of the proposed framework in selecting operating points based on application requirements and desired reliability levels while also emphasizing the increased sensitivity of complex datasets to privacy constraints. {\textit{These results indicate that different trust levels correspond to distinct operating points on the reliability–utility frontier, enabling controlled adjustment of adversarial risk.}}

\section{Conclusion}
This work introduced {TADP-RME}, a trust-adaptive differential privacy framework that overcomes the limitations of fixed-budget noise mechanisms.
We integrate a continuous trust-based privacy budget to enable flexible, interpretable trade-offs between utility and privacy across diverse operating conditions. Unlike conventional DP methods that rely solely on stochastic perturbation, the proposed framework addresses structural leakage through Reverse Manifold Embedding. This further enhances resistance to membership, attribute, and reconstruction attacks, thereby reducing adversarial failure probability. Theoretical analysis affirms the capability of the approach to preserve $(\varepsilon,\delta)$-differential privacy guarantees via post-processing, while introducing additional robustness through structural distortion. Empirical evaluation shows that TADP-RME strikes a balance between utility and reliability across multiple datasets. 

Future work will explore extending the approach to deep neural architectures, learning data-driven transformations, and applying the framework in dynamic real-world environments with evolving trust requirements. This work establishes a direct connection between privacy preservation and system reliability by interpreting adversarial inference as a failure event and demonstrating that structural and stochastic mechanisms can jointly reduce failure probability.

\bibliographystyle{unsrt}  
\bibliography{ref}  

@inproceedings{dwork2006,
author = {Dwork, Cynthia and McSherry, Frank and Nissim, Kobbi and Smith, Adam},
title = {Calibrating noise to sensitivity in private data analysis},
year = {2006},
booktitle = {Proceedings of the Third Conference on Theory of Cryptography},
pages = {265–284},
numpages = {20},
}

@article{dwork2014book,
author = {Dwork, Cynthia and Roth, Aaron},
title = {The Algorithmic Foundations of Differential Privacy},
year = {2014},
issue_date = {Aug 2014},
publisher = {Now Publishers Inc.},
address = {Hanover, MA, USA},
volume = {9},
number = {3–4},
journal = {Found. Trends Theor. Comput. Sci.},
pages = {211–407},
numpages = {197}
}

@ARTICLE{data-driven1,
  author={Yang, Yan and Yu, Juan and Yang, Zhifang and Wang, Guoyin and Yu, Hong and Cheng, Qi},
  journal={IEEE Systems Journal}, 
  title={A Trustable Data-Driven Framework for Composite System Reliability Evaluation}, 
  year={2022},
  volume={16},
  number={4},
  pages={6697-6707},
}

@ARTICLE{privacy-rel1,
  author={Ghiasi, Mohammad and Fotuhi-Firuzabad, Mahmud},
  journal={IEEE Transactions on Reliability}, 
  title={Resilience Enhancement of Smart Power Systems Against False Data Injection Attacks Using Adaptive Intrusion Detection Mechanisms}, 
  year={2025},
  volume={},
  number={},
  pages={1-11},
  }

@article{privacy-rel2,
author = {Qiu, Yuan and Yi, Ke},
title = {Approximate DBSCAN under Differential Privacy},
year = {2025},
publisher = {Association for Computing Machinery},
address = {New York, NY, USA},
volume = {3},
number = {3},
journal = {Proc. ACM Manag. Data},
numpages = {24},

}

@ARTICLE{data-driven2,
  author={Shieh, Shiuhpyng Winston and Voas, Jeff and Laplante, Phil and Rupe, Jason and Hansen, Christian and Wu, Yu-Sung and Chen, Yi-Ting and Li, Chi-Yu and Wu, Kai-Chiang},
  journal={IEEE Transactions on Reliability}, 
  title={Reliability Engineering in a Time of Rapidly Converging Technologies}, 
  year={2024},
  volume={73},
  number={1},
  pages={73-82},
  }

@INPROCEEDINGS{jorgensen2015pdp,
  author={Jorgensen, Zach and Yu, Ting and Cormode, Graham},
  booktitle={2015 IEEE 31st International Conference on Data Engineering}, 
  title={Conservative or liberal? Personalized differential privacy}, 
  year={2015},
  volume={},
  number={},
  pages={1023-1034},
}

@inproceedings{ebadi2016pdp,
author = {Ebadi, Hamid and Sands, David and Schneider, Gerardo},
title = {Differential Privacy: Now it's Getting Personal},
year = {2015},
publisher = {Association for Computing Machinery},
booktitle = {Proceedings of the 42nd Annual ACM SIGPLAN-SIGACT Symposium on Principles of Programming Languages},
pages = {69–81},
numpages = {13},

}

@book{aggarwal2008privacy,
  title={Privacy-Preserving Data Mining: Models and Algorithms},
  author={Aggarwal, Charu C. and Yu, Philip S.},
  publisher={Springer},
  year={2008}
}

@inproceedings{liu2019random,
  title={Privacy-Preserving Data Publishing via Random Projection},
  author={Liu, Kun and Kargupta, Hillol},
  booktitle={Proceedings of the SIAM International Conference on Data Mining (SDM)},
  year={2019}
}

@inproceedings{shokri2017membership,
  title={Membership Inference Attacks Against Machine Learning Models},
  author={Shokri, Reza and Stronati, Marco and Song, Congzheng and Shmatikov, Vitaly},
  booktitle={IEEE Symposium on Security and Privacy (S\&P)},
  pages={3--18},
  year={2017}
}

@inproceedings{carlini2021extracting,
  title={Extracting Training Data from Large Language Models},
  author={Carlini, Nicholas and Tramer, Florian and Wallace, Eric and Jagielski, Matthew and Herbert-Voss, Ariel and Lee, Katherine and Roberts, Adam and Brown, Tom and Song, Dawn and Erlingsson, {\'U}lfar and others},
  booktitle={USENIX Security Symposium},
  pages={2633--2650},
  year={2021}
}

@inproceedings{jayaraman2020evaluating,
  title={Evaluating Differential Privacy in Machine Learning},
  author={Jayaraman, Bharath and Evans, David},
  booktitle={USENIX Security Symposium},
  pages={1895--1912},
  year={2020}
}

@article{cummings2023dp,
  title={Advancing Differential Privacy: Where We Are Now and Future Directions for Real-World Deployment},
  author={Cummings, Rachel and Desfontaines, Damien and Evans, David and Geambasu, Roxana and others},
  journal={arXiv preprint arXiv:2304.06929},
  year={2023}
}

@article{dong2022gdp,
  title={Gaussian Differential Privacy},
  author={Dong, Jinshuo and Roth, Aaron and Su, Weijie},
  journal={Journal of the Royal Statistical Society: Series B (JRSSB)},
  volume={84},
  number={1},
  pages={3--37},
  year={2022}
}

@inproceedings{abadi2016dpsgd,
  title={Deep Learning with Differential Privacy},
  author={Abadi, Martin and Chu, Andy and Goodfellow, Ian and McMahan, Brendan and Mironov, Ilya and Talwar, Kunal and Zhang, Li},
  booktitle={ACM Conference on Computer and Communications Security (CCS)},
  pages={308--318},
  year={2016}
}

@inproceedings{indyk1998lsh,
  title={Approximate Nearest Neighbors: Towards Removing the Curse of Dimensionality},
  author={Indyk, Piotr and Motwani, Rajeev},
  booktitle={ACM Symposium on Theory of Computing (STOC)},
  pages={604--613},
  year={1998}
}

@article{bingham2001random,
  title={Random Projection in Dimensionality Reduction: Applications to Image and Text Data},
  author={Bingham, Ella and Mannila, Heikki},
  journal={Proceedings of the Seventh ACM SIGKDD International Conference on Knowledge Discovery and Data Mining},
  pages={245--250},
  year={2001}
}

@inproceedings{fredrikson2015model,
author = {Fredrikson, Matt and Jha, Somesh and Ristenpart, Thomas},
title = {Model Inversion Attacks that Exploit Confidence Information and Basic Countermeasures},
year = {2015},
isbn = {9781450338325},
publisher = {Association for Computing Machinery},
booktitle = {Proceedings of the 22nd ACM SIGSAC Conference on Computer and Communications Security},
pages = {1322–1333},
numpages = {12},
}

@inproceedings{jayaraman2019evaluating,
author = {Jayaraman, Bargav and Evans, David},
title = {Evaluating differentially private machine learning in practice},
year = {2019},
isbn = {9781939133069},
publisher = {USENIX Association},
booktitle = {Proceedings of the 28th USENIX Conference on Security Symposium},
pages = {1895–1912},
numpages = {18},
}

@inproceedings{nasr2019comprehensive,
  title={Comprehensive Privacy Analysis of Deep Learning: Passive and Active White-box Inference Attacks against Centralized and Federated Learning},
  author={Nasr, Milad and Shokri, Reza and Houmansadr, Amir},
  booktitle={2019 IEEE Symposium on Security and Privacy (SP)},
  pages={739--753},
  year={2019}
}

@book{cover2006elements,
  title={Elements of Information Theory},
  author={Cover, Thomas M. and Thomas, Joy A.},
  year={2006},
  publisher={Wiley}
}

@article{spearman1904,
  title={The proof and measurement of association between two things},
  author={Spearman, Charles},
  journal={The American Journal of Psychology},
  year={1904}
}


\end{document}